\documentclass[a4paper,12pt]{article}

\usepackage{mathrsfs}
\usepackage{graphicx}
\usepackage{bm}
\usepackage{amssymb}
\usepackage{amsmath}
\usepackage{ascmac}
\usepackage{array}

\usepackage{cite}

\usepackage{theorem} 

\theorembodyfont{\rmfamily}

\newtheorem{theorem}{Theorem}[section]
\newtheorem{lemma}{Lemma}[section]
\newtheorem{example}{Example}[section]
\newtheorem{remark}{Remark}[section]
\newtheorem{definition}{Definition}[section]
\newtheorem{proof}{Proof.}
\newtheorem{corollary}{Corollary}[section]

\begin{document}


\title{Second Law-Like Inequalities with Quantum Relative Entropy: An Introduction}
\author{Takahiro Sagawa
\footnote{\scriptsize The Hakubi Center, Kyoto University, Yoshida-ushinomiya cho, Sakyo-ku, Kyoto, 606-8302, Japan}
\footnote{\scriptsize Yukawa Institute for Theoretical Physics, Kyoto University, Kitashirakawa-oiwake cho, Sakyo-ku, Kyoto, 606-8502, Japan}
\footnote{\scriptsize Present address for Ver.~4: Department of Applied Physics, The University of Tokyo, 7-3-1 Hongo, Bunkyo-ku, Tokyo 113-8656, Japan}
}

\date{February 5, 2012\footnote{Several corrections have been made after the first submission; The final version (Ver.~4) has been uploaded on April 17, 2023.}}

\maketitle

\begin{abstract}
We review the fundamental properties of the quantum relative entropy for finite-dimensional Hilbert spaces.  In particular, we focus on several inequalities that are related to the second law of thermodynamics, where the positivity and the monotonicity of the quantum relative entropy play key roles; these properties are directly applicable to derivations of the second law (e.g., the Clausius inequality).  Moreover, the positivity is closely related to the quantum fluctuation theorem, while the monotonicity leads to a quantum version of the Hatano-Sasa inequality for nonequilibrium steady states.  On the basis of the monotonicity, we also discuss  the data processing inequality for the quantum mutual information, which has a similar mathematical structure to that of the second law.  Moreover, we derive a generalized second law with quantum feedback control.  In addition, we review a proof of the monotonicity in line with Petz~\cite{Petz3}.
\end{abstract}

\tableofcontents

\section{Introduction}
The quantum relative entropy~\cite{Wehrl,Ohya} was introduced by Umegaki~\cite{Umegaki}, analogous to the classical relative entropy introduced by Kullback and Leibler~\cite{Kullback}. 
In the the early days~\cite{Lamford,Araki1}, a prominent result on the quantum relative entropy was a proof of the monotonicity~\cite{Lindblad2,Lindblad3,Uhlmann}, which is based on the celebrated Lieb's theorem~\cite{Lieb1} and its application to the strong subadditivity of the von Neumann entropy~\cite{Lieb2,Lieb3}.
Ever since, the quantum relative entropy has been widely applied to quantum information theory~\cite{Ohya2,Nielsen-Chuang,Vedral,Hayashi,Petz1}.  Moreover, the quantum relative entropy is useful to describe the mathematical structure of the second law of thermodynamics~\cite{Callen}.  In fact, it has been shown that the quantum relative entropy is closely related to recent progresses in nonequilibrium statistical mechanics.

The fluctuation theorem is a remarkable result in modern nonequilibrium statistical mechanics, which characterizes a symmetry of fluctuations of the entropy production in thermodynamic systems~\cite{Evans,Gallavotti,Jarzynski1,Kurchan0,Crooks1,Crooks2,Maes1,Jarzynski2,Seifert2,Kawai,Marin,Parrondo,Vaikuntanathan,Kurchan,Tasaki1,Yukawa,Mukamel,Jarzynski3,Roeck,Monnai,Esposito1,Talkner1,Talkner2,Esposito2,Saito,Gaspard,Huber,Utsumi1,Utsumi2,Esposito3,Andireux,Hanggi1,Hanggi2,Nakamura,Ohzeki,Campisi,Lutz,Horowitz1,Horowitz4}. 
From the  fluctuation theorem, we can straightforwardly derive the second law of thermodynamics, which states that the entropy production is nonnegative on average.
It has been understood  that the derivation of the second law based on the fluctuation theorem  is closely related to the positivity of the relative entropy~\cite{Kawai,Marin,Parrondo,Vaikuntanathan,Tasaki1,Esposito3}.  
In fact, without invoking the fluctuation theorem, we can directly derive the second law by using the positivity of the relative entropy~\cite{Tasaki1,Esposito4,Hasegawa}.
If the microscopic dynamics of a thermodynamic system is explicitly described by  quantum mechanics, the fluctuation theorem is referred to as the quantum fluctuation theorem~\cite{Tasaki1,Yukawa,Mukamel,Jarzynski3,Roeck,Monnai,Esposito1,Talkner1,Talkner2,Esposito2,Saito,Gaspard,Huber,Utsumi1,Utsumi2,Esposito3,Andireux,Hanggi1,Hanggi2,Nakamura,Ohzeki,Campisi,Lutz,Horowitz1,Horowitz4}, which is one of the main topics of this article.

On the other hand, the monotonicity implies that  the quantum relative entropy is non-increasing under any time evolution that occurs with unit probability in quantum open systems~\cite{Wehrl,Ohya,Nielsen-Chuang,Hayashi,Petz1,Ruskai,Plenio}.  It has been known that  the second law of thermodynamics can also be derived from the monotonicity of the classical or quantum relative entropy~\cite{Spohn} (see also~\cite{Bergmann,Qian,Yukawa2,Breuer}).
Moreover, the monotonicity can be applied to describe transitions between nonequilibrium steady states (NESSs).   In such situations, the monotonicity leads to a second law-like inequality, which we refer to as a quantum version of the Hatano-Sasa inequality~\cite{Yukawa2}. 
In the classical regime, such an  inequality can also be derived from a generalized fluctuation theorem called the Hatano-Sasa equality.
 It was first discussed for an overdamped Langevin system~\cite{Hatano}, and has been applied  to other situations and systems~\cite{Speck,Sasa,Esposito5,Esposito6,Kurchan2}. 
The quantum Hatano-Sasa equality has been discussed in Refs.~\cite{Lutz,Horowitz4}.

The monotonicity of the quantum relative entropy is also useful in quantum information theory~\cite{Nielsen-Chuang,Hayashi,Petz1}.  In particular, the data processing inequality for the quantum mutual information can be directly derived from the monotonicity as is the case for the second law of thermodynamics.  Moreover, on the basis of the monotonicity, we can derive several important inequalities  such as the Holevo bound, which identifies the upper bound of the accessible classical information that is encoded in quantum states~\cite{Holevo,Yuen_Ozawa,Fuchs_Caves}.  On the other hand, a ``dual'' inequality of the Holevo bound is related to a quantity which we refer to as the QC-mutual information  (or the Groenewold-Ozawa information)~\cite{Groenewold,Ozawa1,Buscemi,Sagawa-Ueda1}.  It  plays a key role in the formulation of a generalized second law of thermodynamics with quantum feedback control~\cite{Lloyd1,Lloyd2,Nielsen,Zurek1,Kieu,Allahverdyan,Quan1,Sagawa-Ueda1,Jacobs,SWKim,Dong,Morikuni,Abreu2,Lahiri,Lu,Funo0,Funo,Park}.

This article aims to be an introduction to the quantum relative entropy for finite-dimensional Hilbert spaces, which is organized as follows.

In Sec.~2, we review the basic properties of quantum mechanics.   In Sec.~2.1, we introduce quantum states and observables.  In Sec.~2.2, we discuss dynamics of quantum systems, which is the main part of this section.  In particular, we introduce the concept of completely-positive (CP) maps, and prove that any CP map has a Kraus representation in line with the proof by  Choi~\cite{Choi1}.

In Sec.~3, we introduce the von Neumann entropy and the quantum relative entropy.  In particular,  we discuss the monotonicity of the quantum relative entropy under completely-positive and trace-preserving (CPTP) maps in Sec.~3.3.  The strong subadditivity of the von Neumann entropy is shown to be a straightforward consequence of the monotonicity of the quantum relative entropy.

In Sec.~4, we discuss the quantum mutual information and related quantities.  In Sec.~4.1, we introduce the quantum mutual information, and prove the data processing inequality on the basis of the monotonicity of the quantum relative entropy.  In Sec.~4.2, we discuss the Holevo's $\chi$-quantity, and prove the Holevo bound.  In Sec.~4.3, we introduce the QC-mutual information and discuss its properties.  In particular, we prove  a ``dual'' inequality of the Holevo bound. We note that all of the three quantities reduce to the classical mutual information for classical cases.

In Sec.~5, we discuss some derivations of the second law of thermodynamics and its variants.  In Sec.~5.1, we discuss the relationship between thermodynamic entropy and the von Neumann entropy.
In Sec.~5.2, we discuss a derivation based on the positivity of the quantum relative entropy.  In Sec.~5.3, we discuss the quantum fluctuation theorem in a general setup, by introducing the stochastic entropy production.  The quantum fluctuation theorem directly leads to the second law, which is  equivalent to the derivation based on the positivity of the quantum relative entropy.  
In Sec.~5.4, we discuss a derivation of the second law based on the monotonicity.  We also discuss relaxation processes toward NESSs, and derive a quantum version of the Hatano-Sasa inequality for transitions between NESSs.

In Sec.~6, we discuss derivations of a generalized second law of thermodynamics with quantum feedback control, which involves the QC-mutual information.  Our derivation is based on the positivity of the quantum relative entropy in Sec.~6.1, while on the monotonicity in  Sec.~6.2.

In Sec.~7, as concluding remarks, we discuss the physical meanings and the validities of the foregoing derivations of the second law of thermodynamics in detail.

In Appendix~A, we briefly summarize the basic concepts in the linear algebra.  In Appendix~B, we prove the monotonicity of the quantum relative entropy in line with  Petz~\cite{Petz3}.


\section{Quantum States and Dynamics}

First of all, we review basic concepts in quantum mechanics. 
We consider quantum systems described by finite-dimensional Hilbert spaces.
  Let $L (\bm H, \bm H')$ be the set of linear operators from  Hilbert space $\bm H$ to $\bm H'$.  In particular, we write $L(\bm H) := L(\bm H, \bm H)$ (see also Table~1). 
 We note that  the linear algebra with the bra-ket notation is briefly summarized in Appendix~A.


\subsection{Quantum States and Observables}

In quantum mechanics, both quantum states and observables (physical quantities) can be described by operators on a Hilbert space.
Let $\bm H$ be a Hilbert space that characterizes a quantum system.  
We define $Q(\bm H) \subset L(\bm H)$ such that  any $\rho \in Q(\bm H)$ satisfies 
\begin{equation}
\rho \geq 0 \ \ \  {\rm and} \ \ \ {\rm tr} [\rho] = 1,
\end{equation}  
where $\rho \geq 0$ means that $\langle \psi | \rho | \psi \rangle \geq 0$ for any $| \psi \rangle \in \bm H$ (or equivalently, $\rho$ is positive),\footnote{In the present article, we follow the terminologies in Ref.~\cite{Nielsen-Chuang} and say that a hermitian operator $X$ is positive if $X \ge 0$ and positive definite if $X>0$.  See appendix A for details.} and ${\rm tr}[\rho] $ means the trace of $\rho$. 
We call $\rho \in Q(\bm H)$  a density operator, which describes a quantum state.  If the rank of $\rho$ is one so that $\rho = | \psi \rangle \langle \psi |$ for $| \psi \rangle \in \bm H$, $\rho$ is called a pure state and $| \psi \rangle$ is called a state vector.\footnote{In terms of the algebraic quantum theory using $C^\ast$-algebras, the definition of pure states depends on the choice of the algebra that is generated by observables.  The above definition is valid only if the set of observables equals the set of all Hermitian operators in $L(\bm H)$.  In other words, we assumed that there is no superselection rule.}  We note that any state vector satisfies $\langle \psi | \psi \rangle = 1$.

\begin{table}
\begin{center}
{\begin{tabular}{cc}
Symbols & Meanings \\
\hline
$L(\bm H, \bm H')$ & The set of linear operators from $\bm H$ to $\bm H'$ \\
$L(\bm H)$ & $L(\bm H, \bm H)$  \\
$Q(\bm H)$ & $\{ \rho \in L(\bm H) : \rho \geq 0 , \ {\rm tr}[\rho] = 1 \}$
\end{tabular}}
\caption{Symbols and their meanings.}
\end{center}
\end{table}

On the other hand, any Hermitian operator $X \in L(\bm H)$ is called an observable, which describes a physical quantity such as a component of the spin of an atom.
Any observable $X$ is assumed to be measurable without any error in principle. 
Let $X := \sum_k x_k | \varphi_k \rangle \langle \varphi_k |$ be the spectrum decomposition, where $\{ | \varphi_k \rangle \}$ is an orthonormal basis of $\bm H$.
By the error-free measurement of observable $X$, the measurement outcome is given by one of $x_k$'s.  For simplicity, we assume that $x_k \neq x_{k'}$ for $k \neq k'$; this assumption will be removed in Sec.~2.2.3.
We note that $k$ is also referred to as an outcome.  The probability of obtaining $x_k$ is given by
\begin{equation}
p(k) := \langle \varphi_k | \rho | \varphi_k \rangle,
\label{Born}
\end{equation}  
where $\rho \in Q(\bm H)$ is the density operator of the measured quantum state.  Equality~(\ref{Born}) is called the Born rule.
The sum of the probabilities satisfies
\begin{equation}
\sum_k p(k) = \sum_k \langle \varphi_k | \rho | \varphi_k \rangle = {\rm tr}[\rho] = 1.
\end{equation}
We note that $\rho \geq 0$ and  ${\rm tr}[\rho] = 1$ respectively   confirm  $p(k) \geq 0$ and  $\sum_k p(k) = 1$.
The average of the outcomes is then given by
\begin{equation}
\langle X \rangle := \sum_k p(k) x_k = {\rm tr}[X\rho],
\label{Born_average}
\end{equation} 
which is a useful formula.
If $\rho = | \psi \rangle \langle \psi |$ is a pure state, the Born rule (\ref{Born}) reduces to
\begin{equation}
p(k) = | \langle \varphi_k | \psi \rangle|^2,
\end{equation}
and Eq.~(\ref{Born_average}) to
\begin{equation}
\langle X \rangle = \langle \psi | X | \psi \rangle.
\end{equation}

\

If there are two quantum systems $A$ and $B$ described by Hilbert spaces $\bm H_A$ and $\bm H_B$, their composite system $AB$ is described by the tensor product space $\bm H_A \otimes \bm H_B$.   
Let $\rho^{AB} \in Q(\bm H_A \otimes \bm H_B)$ be a density operator of the composite system.  The partial states corresponding to $A$ and $B$ are respectively given by
\begin{equation}
\rho^A = {\rm tr}_B [\rho^{AB}], \ \rho^B = {\rm tr}_A [\rho^{AB}],
\label{states}
\end{equation}  
where ${\rm tr}_A$ and ${\rm tr}_B$ describe the partial traces in $\bm H_A$ and $\bm H_B$, respectively.
For any observable $X^A \in L(\bm H_A)$ and the identity $I^B \in L(\bm H_B)$, we have
\begin{equation}
{\rm tr}_{AB} [(X^A\otimes I^B ) \rho^{AB}] = {\rm tr}_A[X^A \rho^A],
\end{equation}
which is consistent with Eq.~(\ref{states}).
If $\rho^{AB} = \rho^A \otimes \rho^B$ is satisfied, $\rho^{AB}$ is called a product state.  
If a pure state is not a product state, it is called an entangled state.

We next show that any quantum state can be written as a pure state of an extended system including an auxiliary system.
Let $K$ be a set of indexes and  $\rho = \sum_{k \in K} p_k | \psi_k \rangle \langle\psi_k | \in Q(\bm H)$ be a state, where $| \psi_k \rangle$'s are not necessarily mutually-orthogonal.  We introduce an auxiliary system $R$ described by a Hilbert space $\bm H_R$ with an orthonormal basis $\{ | r_k \rangle \}_{k \in K}$.  
By defining a pure state
\begin{equation}
| \Psi \rangle := \sum_{k \in K} \sqrt{p_k} | \psi_k \rangle | r_k \rangle \in \bm H \otimes \bm H_R,
\label{purification}
\end{equation}
we have
\begin{equation}
\rho = {\rm tr}_R [| \Psi \rangle \langle \Psi |].
\end{equation}
The state vector $| \Psi \rangle$ in Eq.~(\ref{purification}) is called a purification of $\rho$.  We note that a purification is not unique.


\subsection{Quantum Dynamics}

\subsubsection{Unitary Evolution}

The time evolution of an isolated quantum system is given by a unitary evolution.  A density operator $\rho \in Q(\bm H)$ evolves as 
\begin{equation}
\rho \mapsto U \rho U^\dagger,
\end{equation}
where $U \in L(\bm H)$ is a unitary operator satisfying $U^\dagger U =UU^\dagger = I$.  Any unitary evolution preserves the trace (i.e., ${\rm tr}[U \rho U^\dagger] = {\rm tr}[\rho]$) and  the positivity of $\rho$.
In the continuous-time picture, the time evolution of a density operator is given by the von Neumann equation:
\begin{equation}
\frac{d\rho(t)}{dt} = - {\rm i} [H, \rho(t)] :=  - {\rm i} (H \rho(t) - \rho(t) H),
\end{equation}
where $H \in L(\bm H)$ is a Hermitian operator called the Hamiltonian of the system.  We set $\hbar = 1$ in this article.  The unitary evolution from time $0$ to $t$ is described by  unitary operator $U (t) = e^{-{\rm i}Ht}$.
If we control the system by changing some external classical parameters  such as a magnetic field, the Hamiltonian of the system can depend on time.  In such a case, the von Neumann equation is given by 
\begin{equation}
\frac{d\rho(t)}{dt} = - {\rm i} [H(t), \rho(t)],
\end{equation}
which leads to
\begin{equation}
\begin{split}
U(t) &= \sum_{n=0}^\infty (- {\rm i})^n \int_0^t dt_1 \int_0^{t_1} dt_2 \cdots \int_0^{t_{n-1}} dt_n H(t_1) H(t_2) \cdots H(t_n) \\
&=: {\rm T}\exp \left( -{\rm i} \int_0^t H(t')  dt' \right),
\end{split}
\end{equation}
where ``T'' represents the time-ordered product.

\subsubsection{Completely Positive Maps}

For an open quantum system that interacts with another quantum system, the time evolution is not given by a unitary evolution in general.  
Moreover, the input state (the initial state) and the output state (the final state) can be described by different Hilbert spaces with each other.
The time evolution is generally described by a liner map $\mathcal E : L(\bm H) \to L(\bm H')$, where $\bm H$ and $\bm H'$ are the Hilbert spaces that describe the input and output systems, respectively.

For example, we consider a time evolution in which the input state is in $Q(\bm H)$ and the output state is in $Q(\bm H')$.  If another quantum system described by $\bm H''$ comes to interact with the input system and we access the total output state, then the output system becomes larger than the input one; the output system is described by $\bm H' = \bm H \otimes \bm H''$.  
 
 The general condition for $\mathcal E$ is given by the complete positivity.  If $\mathcal E$ occurs with unit probability, it needs to be trace-preserving.  We will discuss these two properties in detail.

\begin{definition}
$\mathcal E$ is called positive if $\mathcal E (X) \geq 0$ for any $X \in L(\bm H)$ such that $X \geq 0$.  Moreover,  $\mathcal E$ is called completely positive (CP) if $\mathcal E \otimes \mathcal I_n$ is positive for any $n \in \mathbb N$, where $\mathcal I_n$ is the identity operator on $L(\mathbb C^n)$.  
\end{definition}

The positivity is enough to confirm the positivity of the output state in $Q(\bm H')$. 
On the other hand, the complete positivity confirms the positivity of the density operator of the total system including the environment.
In fact, if a time evolution was not CP, one might observe a negative probability in the total system.
 We note that  an important example of $\mathcal E$ that is positive but not CP is the transposition of operators with a matrix representation, which has been used for characterizing entanglements~\cite{Peres,Plenio1}. 

 The definition of the complete positivity can be rewritten as follows:  $\mathcal E$ is CP if $L(\bm H')$-valued matrix $ ( \mathcal E (X_{kl}) )_{1 \leq k,l \leq n}$ is positive  for any $L(\bm H)$-valued positive matrix $X_n := ( X_{kl} )_{1 \leq k,l \leq n}$ with  $X_{kl} \in L(\bm H)$.  The equivalence of the two definitions is confirmed  as follows:  $X_n$ can be written as 
\begin{equation}
X_n = \sum_{kl} X_{kl} \otimes | e_k \rangle \langle e_l | \in L(\bm H \otimes \mathbb C^n),
\end{equation}
where $\{ | e_k \rangle \}_{k=1}^n$ is an orthonormal basis of $\mathbb C^n$.  We then have 
\begin{equation}
(\mathcal E \otimes \mathcal I_n) (X_n) = \sum_{kl} \mathcal E (X_{kl}) \otimes | e_k \rangle \langle e_l | \in L(\bm H' \otimes \mathbb C^n).
\end{equation}
We note that $\mathcal E$ is called  $n$-positive  if $\mathcal E \otimes \mathcal I_n$ is positive for a $n$.

Any time evolution in quantum systems needs to be CP.  On the other hand, we introduce the second important property of  $\mathcal E$:
 
 \begin{definition}
 We call $\mathcal E  : L(\bm H) \to L(\bm H')$ trace-preserving (TP) if ${\rm tr} [\mathcal E (X)]  = {\rm tr}[X]$ for any $X \in L(\bm H)$.
 \end{definition}
 
This property  confirms the conservation of the probability as ${\rm tr} [\mathcal E (\rho)] = 1$ for $\rho \in Q(\bm H)$.
If $\mathcal E$ is both CP and TP, it is called CPTP (completely positive and trace-preserving). 
Any time evolution that occurs with unit probability needs to be CPTP.

 The followings are simple examples.
 
 \begin{example}
 Any unitary evolution $\mathcal E (\rho) := U \rho U^\dagger$  is CPTP.
 \end{example}
 
\begin{example}
Let $\rho \in L(\bm H_A \otimes \bm H_B)$.  The partial trace $\mathcal E (\rho) := {\rm tr}_B [\rho]$ is CPTP.  In fact, $\mathcal E$ is CP, because, for any positive operator $\sigma \in L(\bm H_A \otimes \bm H_B \otimes \mathbb C^n)$, $(\mathcal E \otimes \mathcal I_n) (\sigma) = {\rm tr}_B [\sigma]$ is also positive.   $\mathcal E$ is  TP, because ${\rm tr}_A [{\rm tr}_B [\rho]] = {\rm tr}_{AB} [\rho]$.
\end{example}
 
 \begin{example}
 Let $\rho_A \in Q(\bm H_A)$.  A map $\mathcal E : Q(\bm H) \to Q(\bm H \otimes \bm H_A)$ defined by $\mathcal E(\rho) := \rho \otimes \rho_A$ is CPTP.
\end{example}

We have the following example by combining the above three examples.
 
 \begin{example}
 Let $\rho_A \in Q(\bm H_A)$ be a state  and $U \in Q(\bm H \otimes \bm H_A)$ be a unitary operator.  A map $\mathcal E : Q(\bm H) \to Q(\bm H)$ defined by $\mathcal E (\rho) := {\rm tr}_A [U \rho \otimes \rho_A U^\dagger]$ is CPTP.
 \end{example}
 
 The above example is a typical description of the dynamics in  open quantum systems.  It will be shown in Sec.~2.2.4 that any CPTP map from $L(\bm H)$ to $L(\bm H)$ can be written in this form.  The next example is a generalization of the above example to the cases that the input system and the output system are different.
 
 \begin{example}
 Let $\bm H$ and $\bm H'$ be Hilbert spaces corresponding to the input and the output systems, respectively.  We introduce auxiliary systems described by $\bm H_A$ and $\bm H_B$ and assume that $\bm H \otimes \bm H_A \simeq \bm H' \otimes \bm H_B$.  Let $\rho_A \in Q(\bm H_A)$ be a state and $U \in L(\bm H \otimes \bm H_A )  \simeq L(\bm H' \otimes \bm H_B)$ be a unitary operator.  Then a linear map  $\mathcal E : L(\bm H) \to L(\bm H')$ defined by $\mathcal E(\rho) := {\rm tr}_B [U \rho \otimes \rho_A U^\dagger]$ is CPTP.
 \end{example}
 
We next consider quantum measurement processes~\cite{Neumann,Davies-Lewis,Stinespring,Kraus,Ozawa2,Nielsen-Chuang,Shimizu}.  Suppose that we perform a measurement on a quantum system and obtain outcome $k$ with probability $p(k)$.  We assume that the number of possible outcomes is finite.  Let $\rho \in Q(\bm H)$ be the pre-measurement state and $\rho_k \in Q(\bm H')$ be the post-measurement state with outcome $k$. We define a linear map $\mathcal E_k : Q(\bm H) \to Q(\bm H')$ such that 
\begin{equation}
\rho_k = \frac{1}{ p(k)}\mathcal E_k (\rho),
\end{equation}
where the probability of obtaining outcome $k$ is given by
\begin{equation}
p(k) = {\rm tr} [\mathcal E_k (\rho)].
\label{probability}
\end{equation}
The map $\mathcal E_k$ needs to be CP, but it is not TP if $p(k) \neq 1$.   In general, a time evolution that does not occur with unit probability does not need to be TP.
We note that $\mathcal E_k$ needs to satisfy 
\begin{equation}
{\rm tr}[\mathcal E_k (X) ] \leq {\rm tr}[X]
\end{equation}
for any positive $X \in L(\bm H)$, since $p(k) \leq 1$ must hold for any $\rho \in Q(\bm H)$.  The ensemble average  of $\rho_k$'s over all outcomes is given by $\sum_k p(k) \rho_k = \sum_k \mathcal E_k (\rho) =: \mathcal E (\rho)$, where  $\mathcal E := \sum_k \mathcal E_k$ is a CPTP map. We note that $\{ \mathcal E_k \}$ is called an instrument, which characterizes the measurement process.\footnote{Rigorously speaking, an instrument is a map from $K'$ to $\sum_{k \in K'} \mathcal E_k$, where $K'$ is an element of  a $\sigma$-algebra over $K = \{ k \}$.}

We note that $\mathcal E : L(\bm H) \to L(\bm H')$ is called a unital map if it satisfies $\mathcal E(I) = I'$, where $I$ and $I'$ are the identities on $\bm H$ and $\bm H'$, respectively.

 \subsubsection{Kraus Representation}

We next show that any CP map has a useful representation, which is called the Kraus representation.

\begin{theorem}[Kraus representation]
A linear map  $\mathcal E: L(\bm H) \to L(\bm H')$ is CP if and only if it can be written as
\begin{equation}
\mathcal E (\rho) = \sum_k M_k \rho M_k^\dagger,
\label{Kraus}
\end{equation}
where $\rho \in L(\bm H)$, $M_k \in L(\bm H, \bm H')$, and the sum in the right-hand side (rhs) is taken over a finite number of $k$'s.  Equality~(\ref{Kraus}) is called a Kraus representation, and $M_k$'s are called Kraus operators.
\end{theorem}

\begin{proof}[Choi~\cite{Choi1}] 
(\textit{Proof of $\Rightarrow$}) \
Suppose that $\mathcal E$ is CP.  Let $\{ | e_i \rangle \}_i$ be an orthonormal basis of $\bm H$.
An operator $E := \sum_{ij} | e_i \rangle \langle e_j | \otimes  | e_i \rangle \langle e_j |$ is positive because $\langle \psi | E | \psi \rangle = | \sum_i  \langle e_i | \langle e_i | \psi \rangle |^2 \geq 0$ for any $| \psi \rangle \in \bm H \otimes \bm H$.  Since $\mathcal E$ is CP, $\mathcal E \otimes \mathcal I$ is positive with $\mathcal I$ the identity on $L(\bm H)$. Therefore, the following operator is positive:
\begin{equation}
(\mathcal E \otimes \mathcal I) (E) = \sum_i \mathcal E ( | e_i \rangle \langle e_j | ) \otimes  | e_i \rangle \langle e_j | \in L(\bm H' \otimes \bm H).
\end{equation}
Therefore, it has a spectrum decomposition of the form $(\mathcal E \otimes \mathcal I) (E)  = \sum_k | v_k \rangle \langle v_k |$ with $| v_k \rangle \in \bm H' \otimes \bm H$. We note that  $| v_k \rangle$ can be written as $| v_k \rangle = \sum_i | x_k^i \rangle | e_i \rangle$ with $| x_k^i \rangle \in \bm H'$.  
We then obtain 
\begin{equation}
(\mathcal E \otimes \mathcal I ) (E) = \sum_{kij} |x_k^i \rangle \langle x_k^j | \otimes | e_i \rangle \langle e_j |.
\label{proof1_1}
\end{equation}
In addition, $|x_k^i \rangle \langle x_k^j | = M_k | e_i \rangle \langle e_j | M_k^\dagger$ holds with   $M_k := \sum_i | x_k^i \rangle \langle e_i | \in L(\bm H, \bm H')$.  Therefore, we obtain
\begin{equation}
\mathcal E ( | e_i \rangle \langle e_j |) = \sum_k M_k | e_i \rangle \langle e_j | M_k^\dagger
\end{equation}
 for any $(i,j)$, which implies Eq.~(\ref{Kraus}).
\\ 
(\textit{Proof of $\Leftarrow$}) \ Suppose Eq.~(\ref{Kraus}).  Then $(\mathcal E \otimes \mathcal I_n) (\sigma) = \sum_k (M_k \otimes I_n) \sigma (M_k^\dagger \otimes I_n) $ holds for any $\sigma \in L(\bm H \otimes \mathbb C^n)$, where $\mathcal I_n$ and $I_n$ are the identities on $L(\mathbb C^n)$ and  $\mathbb C^n$, respectively.  We then obtain, for any $| \psi \rangle \in \bm H' \otimes \mathbb C^n$,
\begin{equation}
\langle \psi | (\mathcal E \otimes \mathcal I_n) (\sigma) | \psi  \rangle = \sum_k \langle \psi_k | \sigma | \psi_k \rangle, 
\end{equation}
where $| \psi_k \rangle := M_k^\dagger \otimes I_n | \psi \rangle$.  Therefore, $\langle \psi | (\mathcal E \otimes \mathcal I_n) (\sigma) | \psi  \rangle$ is positive if $\sigma$ is positive, which implies that $\mathcal E \otimes \mathcal I_n$ is positive and therefore $\mathcal E$ is CP.
$\Box$\end{proof}

Theorem~2.1 was first proved by Kraus~\cite{Kraus} based on  Stinespring's theorem~\cite{Stinespring}.  The above proof is based on Choi's proof~\cite{Choi1}.

We note that the Kraus representation is not unique.  We also note that the above proof implies that  the $N$-positivity is enough for $\mathcal E$ to have a Kraus representation with $N$ the dimension of $\bm H$.  The following theorem connects the condition of TP to the Kraus representation.

\begin{theorem}
Let $\mathcal E: L(\bm H) \to L(\bm H')$ be a CP map.  $\mathcal E$ is TP if and only if  Kraus operators satisfy
\begin{equation}
\sum_k M_k^\dagger M_k = I,
\label{Kraus_unity}
\end{equation} 
where $I$ is the identity on $\bm H$.
\end{theorem}

\begin{proof}
Suppose that $\mathcal E$ is CPTP.  We then have
\begin{equation}
{\rm tr}[\rho] = {\rm tr} [\mathcal E (\rho)] = {\rm tr} \left[ \sum_k M_k^\dagger M_k \rho \right]
\end{equation}
for any $\rho \in L(\bm H)$, which implies Eq.~(\ref{Kraus_unity}).  Conversely, if Eq.~(\ref{Kraus_unity}) is satisfied, ${\rm tr}[\rho] = {\rm tr} [\mathcal E (\rho)] $ holds. 
$\Box$\end{proof}

We consider quantum measurements in terms of  the Kraus representation.
We first note that CP map $\mathcal E_k$ is written as 
\begin{equation}
\mathcal E_k (\rho) = \sum_i M_{ki} \rho M_{ki}^\dagger.
\label{Kraus2}
\end{equation}
Since $\mathcal E := \sum_{k} \mathcal E_k$ is assumed to be CPTP, we have
\begin{equation}
\sum_{ki} M_{ki}^\dagger M_{ki} = I.
\end{equation}
The probability~(\ref{probability}) of outcome $k$ is then written as
\begin{equation}
p(k) = {\rm tr} \left[ \sum_i M_{ki} \rho M_{ki}^\dagger \right] = {\rm tr} [E_k \rho],
\label{POVM_probability}
\end{equation}
where we defined
\begin{equation}
E_k := \sum_{i} M_{ki}^\dagger M_{ki}.
\end{equation}
It is obvious that $E_k$'s satisfy
\begin{equation}
E_k \geq 0
\label{POVM_positive}
\end{equation}
and
\begin{equation}
\sum_{k } E_k = I.
\label{POVM_unity}
\end{equation}
Inequality~(\ref{POVM_positive}) confirms the positivity of the probability (i.e., $p(k) \geq 0$), and Eq.~(\ref{POVM_unity}) confirms that $\sum_k p(k) = 1$.  Any $\{ E_k \} \subset L(\bm H)$ satisfying (\ref{POVM_positive}) and (\ref{POVM_unity}) is called the positive operator-valued measure (POVM).\footnote{Rigorously speaking, a POVM is a map from $K'$ to $\sum_{k \in K'} E_k$, where $K'$ is an element of  a $\sigma$-algebra over $K = \{ k \}$.}

The projection measurement of an observable is a special case of the foregoing general formulation of quantum measurements.
Let $X = \sum_k x_k P_k$ be an observable, where $P_k$'s are projection operators. We assume that $x_k \neq x_{k'}$ for $k \neq k'$.
If the Kraus representation of $\mathcal E_k$ is given by
\begin{equation}
\mathcal E_k (\rho) = P_k \rho P_k, 
\end{equation}
then the measurement is called the projection measurement of $X$.  In this case, the POVM  consists of projection operators $\{ P_k \}$, and the probability of outcome $k$ is given by $p(k) = {\rm tr}[P_k \rho]$, which is a slight generalization of the Born rule (\ref{Born}).  We also refer to the projection measurement of $X$ as that of $\{ P_k \}$.  In particular, if $P_k$ is written as $P_k = | \psi_k \rangle \langle \psi_k |$ for any $k$, the projection measurement of $X$ is referred to as that of  orthonormal basis $\{ | \psi_k \rangle \}$.

\begin{example}
We consider a simple model of a photodetection.  Suppose that $\bm H$ is 2-dimensional and describes a two-level atom.  Let $\{ | 0 \rangle , | 1 \rangle \} \subset \bm H$ be an orthonormal basis, where $| 0 \rangle$ and $| 1 \rangle$ respectively describe the ground state and the excited state.  The atom emits a photon with probability $p$ if it is in the excited state.  We observe the photon number with unit efficiency, where the outcome is given by ``$0$'' or ``$1$.''  In this case, the  Kraus operators are given by
\begin{equation}
M_0 := | 0 \rangle \langle 0 | + \sqrt{1-p}| 1 \rangle \langle 1 |, \ M_1 := \sqrt{p} | 0 \rangle \langle 1 |,
\end{equation}
which leads to the POVM that consists of 
\begin{equation}
E_0 := | 0 \rangle \langle 0 | +  (1-p) | 1 \rangle \langle 1 |,  \  E_1 := p | 1 \rangle \langle 1 |. 
\end{equation}
If $p=1$ holds, this measurement becomes the projection measurement of $X := x_0 | 0 \rangle \langle 0 | + x_1 | 1 \rangle \langle 1 |$, where we can define $x_0 := 0$ and $x_1 := 1$.
\end{example}

\subsubsection{Indirect Measurement Model}

We next show that any quantum measurement described by $\{ \mathcal E_k \}$ with $\mathcal E_k : L(\bm H) \to L(\bm H)$ can be written by a simple model of an indirect measurement.    

\begin{theorem}
Let $\{ \mathcal E_k \}$ be an instrument with $\mathcal E_k : L(\bm H) \to L(\bm H)$.  There exist an auxiliary system $R$ described by $\bm H_R$, unitary operator $U \in L(\bm H \otimes \bm H_R)$, a reference state $| \phi_R \rangle \in \bm H_R$, and projection operators $\{ P_k \} \subset \bm H_R$ satisfying $\sum_k P_k = I$, such that
\begin{equation}
\mathcal E_k (\rho) = {\rm tr}_R [(I \otimes P_k) U \rho \otimes | \phi_R \rangle \langle \phi_R | U^\dagger ( I \otimes  P_k) ] 
\label{measurement}
\end{equation}
for any $\rho$ and $k$, where $I$ is the identity on $\bm H$.
\end{theorem}

\begin{proof}
Let $\{ M_{ki}  \}_{i} \subset L(\bm H)$ be  the set of Kraus operators of  $\mathcal E_k$, which is given by Eq.~(\ref{Kraus2}).  
We introduce an auxiliary system such that  $\bm H_R$ has an orthonormal basis $\{ | ki \rangle \}_{ki}$.   
We first show that, for any $| \phi_R \rangle \in \bm H_R$, there is a unitary operator $U \in L(\bm H \otimes \bm H_R)$ satisfying 
\begin{equation}
U | \psi \rangle | \phi_R \rangle = \sum_{ki} M_{ki} | \psi \rangle | ki \rangle.
\end{equation}
In fact, $U$ conserves the inner product on $\{ | \varphi \rangle | \phi_R \rangle ; \ | \varphi \rangle \in \bm H \} \subset \bm H \otimes \bm H_R$, that is,
\begin{equation}
\left( \sum_{k'i'}\langle \varphi | \langle k'i' |  M_{k'i'}^\dagger \right) \left( \sum_{ki} M_{ki} | \psi \rangle | ki \rangle \right) = \sum_{ki} \langle \varphi | M_{ki}^\dagger M_{ki} | \psi \rangle = \langle \varphi | \psi \rangle
\end{equation}
holds for any $|\psi \rangle, | \varphi \rangle \in \bm H$.  We then have
\begin{equation}
U \rho \otimes | \phi_R \rangle \langle \phi_R | U^\dagger =  \sum_{kik'i'} M_{ki} \rho M_{k'i'}^\dagger \otimes | ki \rangle \langle k'i' |.
\end{equation}
By defining  $P_{k} := \sum_{i} | ki \rangle \langle ki | \in L(\bm H_R)$, we have  
\begin{equation}
( I \otimes  P_k) U \rho \otimes | \phi_R \rangle \langle \phi_R | U^\dagger ( I \otimes  P_k)  =  \sum_{ii'} M_{ki} \rho M_{ki'}^\dagger \otimes | ki \rangle \langle ki' |,
\end{equation}
and therefore
\begin{equation}
{\rm tr}_R [ ( I \otimes  P_k) U \rho \otimes | \phi_R \rangle \langle \phi_R | U^\dagger ( I \otimes  P_k) ] = \sum_{i} M_{ki} \rho M_{ki}^\dagger  = \mathcal  E_k (\rho),
\end{equation}
which implies (\ref{measurement}).
$\Box$\end{proof}

Physically, $R$ can be regarded as a probe system such as the local oscillator of a homodyne detection, and   $| \phi_R \rangle$ as the initial state of the probe such as a coherent state.  The measured system described by $\bm H$ interacts with the probe system by the unitary evolution.  We next perform the projection measurement with $\{ P_k \}$ on the probe, and obtain the information about $\rho$.  The effect of this indirect measurement is characterized only by instrument $\{ \mathcal E_k \}$.   In the case of a CPTP map, Theorem~2.3 reduces to the following corollary.

\begin{corollary}
Let $\mathcal E : L(\bm H) \to L(\bm H)$ be a CPTP map.  There exist an auxiliary system $R$ described by $\bm H_R$, unitary operator $U \in L(\bm H \otimes \bm H_R)$, and a reference state $| \phi_R \rangle \in \bm H_R$ such that
\begin{equation}
\mathcal E (\rho) = {\rm tr}_R [ U \rho \otimes | \phi_R \rangle \langle \phi_R | U^\dagger]. 
\end{equation}
\end{corollary}

The above corollary implies that any nonunitary evolution $\mathcal E: L(\bm H) \to L(\bm H)$ can be modeled by a unitary evolution of an extended system.

We note that, in Theorem~2.3 and Corollary~2.1, the initial state of the total system described by $\bm H \otimes \bm H_R$ is a product state.  In fact, if a CP map is reproduced by a single indirect measurement model for an arbitrary input state $\rho \in Q(\bm H)$, then the initial state of $R$ should be independent of $\rho$.

\subsubsection{Heisenberg Picture}

We briefly discuss the Schr\"{o}dinger and  the Heisenberg pictures of time evolutions.
 We define an inner product  of $X, Y \in L(\bm H)$ by
 \begin{equation}
 \langle X,Y \rangle_{\rm HS} := {\rm tr}[X^\dagger Y],
 \label{Hilbert_Schmidt}
 \end{equation}
which is called the Hilbert-Schmidt inner product.  We define the adjoint of a linear map $\mathcal E: L(\bm H) \to L(\bm H')$ in terms of the Hilbert-Schmidt inner product.   A linear map $\mathcal E^\dagger: L(\bm H') \to L( \bm H)$ is the adjoint of $\mathcal E$ if it satisfies 
\begin{equation}
\langle \mathcal E^\dagger (X), Y\rangle_{\rm HS} = \langle X, \mathcal E(Y) \rangle_{\rm HS}
\end{equation}
for any $X \in L(\bm H')$ and $Y \in L(\bm H)$.  If $\mathcal E$ is CP and its  Kraus representation is given by Eq.~(\ref{Kraus}),  the adjoint of $\mathcal E$ is written as
\begin{equation}
\mathcal E^\dagger (X) = \sum_k M_k^\dagger X M_k,
\end{equation}
which implies that $\mathcal E^\dagger$ is also CP from Theorem~2.1.  We note that a CP map $\mathcal E$ is TP if and only if $\mathcal E^\dagger$ is unital, because $\mathcal E^\dagger (I) = \sum_k M_k^\dagger M_k$.

Let $\rho \in Q(\bm H)$ be a state,  $X \in L(\bm H')$ be an observable, and $\mathcal E: L(\bm H) \to L(\bm H')$ be  a CP map.   We then have
\begin{equation}
{\rm tr}[X \mathcal E (\rho)] = {\rm tr} [\mathcal E^\dagger (X) \rho],
\end{equation}
where the left-hand side (lhs) is called the Schr\"{o}dinger picture, while the rhs is called the Heisenberg picture.


\section{Quantum Relative Entropy}

We now introduce the  quantum entropies and discuss their basic properties.


\subsection{Von Neumann Entropy}

We first introduce the von Neumann entropy~\cite{Neumann}.
\begin{definition}
The von Neumann entropy of $\rho \in Q(\bm H)$ is defined as
\begin{equation}
S(\rho) := - {\rm tr} [\rho \ln \rho].
\end{equation}
\end{definition}

\begin{remark}
We note the relationship between the von Neumann entropy and the classical Shannon entropy~\cite{Shannon,Cover-Thomas}.
Let $\{ p(a) \}_{a \in A}$ be a probability distribution on a set $A$.  We regard the distribution as a vector whose $a$-th entry is $p(a)$, which we denote  as  $\bm p := ( p(a) )_{a \in A}$.  The Shannon entropy (or the Shannon information) of $\bm p$ is defined as
\begin{equation}
H(\bm p) := - \sum p(a) \ln p(a).
\end{equation}
If the spectrum decomposition of $\rho$ is given by $\rho = \sum p(a) | \varphi_a \rangle \langle \varphi_a |$ with an orthonormal basis $\{ | \varphi_a \rangle \}_{a \in A}$, the von Neumann entropy of $\rho$ reduces to the Shannon entropy of $\bm p$:
\begin{equation}
S(\rho) = H(\bm p).
\end{equation}
\end{remark}

The following theorems describe basic properties of the von Neumann entropy.

\begin{theorem}
Let $\rho_k$'s are density operators whose supports are mutually orthogonal.  Then the von Neumann entropy of $\rho := \sum_k p_k \rho_k$ with $\sum p_k = 1$ satisfies
\begin{equation}
S(\rho) = H(\bm p) + \sum_k p_k S(\rho_k),
\label{decomposition}
\end{equation} 
where $H(\bm p) := -\sum_k p_k \ln p_k$.
\end{theorem}

\begin{proof}
Since the supports of $\rho_k$'s are mutually orthogonal, we have
\begin{equation}
\begin{split}
S(\rho) &= - \sum_k {\rm tr}[ p_k \rho_k \ln (p_k \rho_k)] \\
&= - \sum_k {\rm tr} [ p_k  \rho_k ( \ln p_k + \ln \rho_k ) ] \\
&=  - \sum_k p_k \ln p_k - \sum_k p_k {\rm tr} [\rho_k \ln \rho_k ], 
\end{split}
\end{equation}
which implies Eq.~(\ref{decomposition}).
$\Box$\end{proof}

\begin{theorem}
Let $| \Psi \rangle \in \bm H_A \otimes \bm H_B$ be a state vector of a composite system, whose partial states are given by $\rho^A := {\rm tr}_B [| \Psi \rangle \langle \Psi |]$ and  $\rho^B := {\rm tr}_A [| \Psi \rangle \langle \Psi |]$.  Then
\begin{equation}
S(\rho^A) = S(\rho^B).
\label{purification_entropy}
\end{equation}
\end{theorem}

\begin{proof}
Let $| \Psi \rangle = \sum_{k} \sqrt{p_k} | \varphi_k \rangle | \psi_k \rangle$ be the Schmidt decomposition of $| \Psi \rangle$ with $\sum_k p_k = 1$.  Then $\rho^A = \sum_k p_k | \varphi_k \rangle \langle \varphi_k |$ and $\rho^B = \sum_k p_k | \psi_k \rangle \langle \psi_k |$ hold, and therefore $S(\rho^A) = S(\rho^B) = -\sum_k p_k \ln p_k$.
$\Box$\end{proof}


\subsection{Quantum Relative Entropy and Its Positivity}

We now introduce the quantum relative entropy and prove its positivity.
\begin{definition}
Let  $\rho, \sigma \in Q(\bm H)$. The Quantum relative entropy of $\rho$ to $\sigma$ is defined as
\begin{equation}
S(\rho \| \sigma) := {\rm tr} [\rho \ln \rho] - {\rm tr} [\rho \ln \sigma].
\end{equation}
If there exists $| \psi \rangle \in \bm H$ that satisfies $\sigma | \psi \rangle = 0$ and $\langle \psi | \rho | \psi \rangle \neq 0$, the quantum relative entropy is defined as $S(\rho \| \sigma) := + \infty$.
\end{definition}

\begin{remark}
We note the relationship between the quantum relative entropy and the classical relative entropy (the Kullback-Leibler divergence)~\cite{Kullback,Cover-Thomas}.
Let $\bm p := ( p(a) )_{a \in A}$ and $\bm q := ( q(a) )_{a \in A}$ be probability distributions on a set $A$.   The classical relative entropy of  $\bm p$ to $\bm q$ is  defined as  
\begin{equation}
S( \bm p \| \bm q ) := \sum_a p(a) \ln \frac{p(a)}{q(a)}.
\end{equation} 
If two density operators are given by $\rho = \sum_a p(a) | \varphi_a \rangle \langle \varphi_a |$ and $\sigma = \sum_a q(a) | \varphi_a \rangle  \langle \varphi_a |$ with an orthonormal basis $\{ | \varphi_a \rangle \}_{a \in A}$, the quantum relative entropy reduces to the classical one:
\begin{equation}
S(\rho \| \sigma) = S(\bm p \| \bm q).
\end{equation}
\end{remark}

We next prove the positivity of the quantum relative entropy, which plays a key role to derive the second law of thermodynamics in Secs.~5.2 and 5.3.

\begin{theorem}[Positivity of  the quantum relative entropy]
\begin{equation}
S(\rho \| \sigma) \geq 0,
\label{positivity}
\end{equation}
where the equality is achieved if and only if $\rho = \sigma$.  Inequality~(\ref{positivity}) is  called the Klein inequality.
\end{theorem}

\begin{proof}
Let $\rho = \sum_a p(a) | \psi_a \rangle \langle \psi_a |$ and $\sigma = \sum_a q(a) | \varphi_a \rangle \langle \varphi_a |$, where $\{ | \psi_a \rangle \}$ and $\{ | \varphi_a \rangle \}$ are orthonormal bases.  
We first show that
\begin{equation}
S(\rho \| \sigma ) \geq S(\rho \| \sigma'),
\label{posi1}
\end{equation}
where $\sigma' := \sum_a | \psi_a \rangle \langle \psi_a | \sigma | \psi_a \rangle \langle \psi_a | = \sum_a  q'(a)   | \psi_a \rangle \langle \psi_a |$ with  $q'(a) := \sum_b q(b) | \langle \psi_a | \varphi_b \rangle |^2$.  Inequality~(\ref{posi1})  is equivalent to $-{\rm tr}[\rho \ln \sigma] \geq - {\rm tr}[\rho \ln \sigma']$.  We note that $-{\rm tr}[\rho \ln \sigma] = -\sum_a p(a) \langle \psi_a | \ln \sigma | \psi_a \rangle$ and $- {\rm tr}[\rho \ln \sigma'] = - \sum_a p(a) \ln q'(a)$.  By applying the Jensen inequality to convex function $- \ln x$, we have
\begin{equation}
\begin{split}
- \langle \psi_a | \ln \sigma | \psi_a \rangle &= - \sum_b | \langle \psi_a | \varphi_b \rangle |^2 \ln q(b) \\
&\geq  - \ln \left( \sum_b  | \langle \psi_a | \varphi_b \rangle |^2 q(b) \right) \\
&= - \ln q'(a),
\end{split}
\end{equation}
where we used that $\sum_b | \langle \psi_a | \varphi_b \rangle |^2 = 1$ holds for any $a$. Therefore, we obtain inequality (\ref{posi1}).  The equality in (\ref{posi1}) is achieved if and only if $| \langle \psi_a | \varphi_{f(b)} \rangle |^2 = \delta_{ab}$, where $f(\cdot)$ is a bijection map and $\delta_{ab}$ is the Kronecker delta.  By relabeling the indexes of $\{ | \varphi_b \rangle \}$, we can choose $f(b) = b$ without loss of generality.

We next show 
\begin{equation}
S(\rho \| \sigma') \geq 0,
\label{posi2}
\end{equation}
which is equivalent to the positivity of the classical relative entropy $\sum_a p(a) \ln (p(a) / q'(a))$.  By using inequality $\ln (x^{-1}) \geq 1-x$ for $x>0$, we obtain
\begin{equation}
\sum_a p(a) \ln \frac{p(a)}{q'(a)} \geq \sum_a p(a) \left( 1- \frac{q'(a)}{p(a)} \right) = 0,
\end{equation}
which implies inequality~(\ref{posi2}).  The equality in (\ref{posi2}) is achieved if and only if $p(a) = q'(a)$ for any $a$.

By combining inequalities (\ref{posi1}) and (\ref{posi2}), we obtain inequality (\ref{positivity}).  The equality is achieved if and only if $| \langle \psi_a | \varphi_{b} \rangle |^2 = \delta_{ab}$ and $p(a) = q'(a)$,  which implies $\rho = \sigma$.
$\Box$\end{proof}

By removing the assumptions that ${\rm tr}[\rho] = \sum_a p(a) = 1$ and ${\rm tr}[\sigma] = \sum_a q(a) = 1$ from the above proof, we can straightforwardly obtain a generalization of Eq.~(\ref{positivity}):
\begin{equation}
{\rm  tr} [\rho (\ln \rho - \ln \sigma )] \geq {\rm tr} [\rho - \sigma]
\label{Klein}
\end{equation} 
for any positive operators $\rho$ and $\sigma$.  Inequality~(\ref{Klein}) is also called the Klein inequality.

The subadditivity of the von Neumann entropy is a direct consequence of the positivity of the quantum relative entropy.

\begin{theorem}[Subadditivity of von Neumann entropy]
Let $\rho \in Q(\bm H_A \otimes \bm H_B)$ be a density operator of a composite system, whose partial states are given by  $\rho^A := {\rm tr}_B[\rho^{AB}] \in Q(\bm H_A)$ and $\rho^B :=  {\rm tr}_A[\rho^{AB}] \in Q(\bm H_B)$.  Their von Neumann entropies satisfy
\begin{equation}
S(\rho^{AB}) \leq S(\rho^A) + S(\rho^B),
\label{subadditivity}
\end{equation}
where the equality is achieved if and only if $\rho^{AB} = \rho^A \otimes \rho^B$.
\end{theorem}

\begin{proof} We have
\begin{equation}
S(\rho^A) + S(\rho^B) - S(\rho^{AB}) = S(\rho^{AB} \| \rho^A \otimes \rho^B) \geq 0,
\end{equation}
where the right equality is achieved if and only if $\rho^{AB} = \rho^A \otimes \rho^B$.
$\Box$\end{proof}


\subsection{Monotonicity of the Quantum Relative Entropy}

We next discuss that the quantum relative entropy is non-increasing under any CPTP map, which is called the monotonicity.  The monotonicity can be applied to a derivation of the second law of thermodynamics and to a lot of theorems in quantum information theory.

\begin{theorem}[Monotonicity of the quantum relative entropy]
Let  $\rho, \sigma \in Q(\bm H)$  be states and $\mathcal E: L(\bm H) \to L(\bm H')$ be a CPTP map.  Then 
\begin{equation}
S(\mathcal E (\rho) \| \mathcal E(\sigma )) \leq S(\rho \| \sigma),
\label{monotonicity}
\end{equation}
which is called the Uhlmann inequality.
\end{theorem}

Several proofs of the monotonicity have been known, which are not so simple.  One of the proofs will be shown in Appendix~B.  The following corollary is a special case of the monotonicity.

\begin{corollary}
Let $\rho^{AB}, \sigma^{AB} \in Q(\bm H_A \otimes \bm H_B)$ be density operators of a composite system, whose partial states are given by $\rho^A := {\rm tr}_B[\rho^{AB}] $ and $\sigma^A := {\rm tr}_B[\sigma^{AB}] $.  Then
\begin{equation}
S(\rho^A \| \sigma^A) \leq S(\rho^{AB} \| \sigma^{AB}).
\label{monotonicity2}
\end{equation}
\end{corollary}

\begin{proof}
Apply the monotonicity (\ref{monotonicity}) to a CPTP map $\mathcal E: L(\bm H_A \otimes \bm H_B) \to L(\bm H_A)$ such that $\mathcal E (\rho) := {\rm tr}_B [\rho]$ for $\rho \in L(\bm H_A \otimes \bm H_B)$.
$\Box$\end{proof}

We also have the following corollary, which implies that the von Neumann entropy is non-decreasing for a special class of CPTP maps.

\begin{corollary}
Let $\mathcal E: L(\bm H) \to L(\bm H)$ be a unital CPTP map satisfying $\mathcal E (I) = I$, where $I$ is the identity on $\bm H$.  The von Neumann entropy is then non-decreasing:
\begin{equation}
S(\rho) \leq S(\mathcal E(\rho)).
\label{entropy_increase}
\end{equation}
\end{corollary}

\begin{proof}
Let $d$ be the dimension of $\bm H$.  We than have
\begin{equation}
\begin{split}
S(\rho) &= -S(\rho \| I/d) + \ln d \leq -S(\mathcal E(\rho) \| \mathcal E(I/d)) + \ln d \\
&= -S(\mathcal E(\rho) \| I/d) + \ln d = S(\mathcal E(\rho)),
\end{split}
\end{equation}
where we used the monotonicity of the quantum relative entropy.
$\Box$\end{proof}

Let $\mathcal E(\rho) = \sum_k M_k \rho M_k^\dagger$ be a Kraus representation of $\mathcal E$.
The condition of $\mathcal E (I) = I$ is satisfied if all of the Kraus operators are Hermitian such that $M_k = M_k^\dagger$.  In particular, $\mathcal E (I) = I$ holds if $M_k$'s are projection operators

The strong subadditivity of the von Neumann entropy is easily obtained from the monotonicity of the quantum relative entropy.

\begin{theorem}[Strong subadditivity of the von Neumann entropy]
Let $\rho^{ABC} \in Q(\bm H_A \otimes \bm H_B \otimes \bm H_C)$,  $\rho^{AB} := {\rm tr}_C[\rho^{ABC}]$,  $\rho^{BC} := {\rm tr}_A[\rho^{ABC}]$, and $\rho^{B} := {\rm tr}_{AC}[\rho^{ABC}] $.  Then their von Neumann entropies satisfy
\begin{equation}
S(\rho^{ABC}) + S(\rho^B) \leq S(\rho^{AB}) + S(\rho^{BC}).
\label{strong_subadditivity}
\end{equation}
\end{theorem}

\begin{proof}
Let $\sigma^A := I_A / d_A$, where $I_A$ is the identity on $\bm H_A$ and $d_A$  is the dimension of $\bm H_A$.  We then have
\begin{equation}
\begin{split}
&{}  [ S(\rho^{AB}) + S(\rho^{BC})] - [ S(\rho^{ABC}) + S(\rho^B) ] \\
&= [ S(\rho^{AB}) -  S(\rho^{ABC})] - [ S(\rho^B) - S(\rho^{BC})] \\
&= S(\rho^{ABC} \| \sigma^A \otimes \rho^{BC}) - S(\rho^{AB} \| \sigma^A \otimes \rho^B) \\
&\geq 0,
\end{split}
\end{equation}
where we used the monotonicity of the relative entropy for a CPTP map $\mathcal E: L(\bm H_A \otimes \bm H_B \otimes \bm H_C) \to L(\bm H_A \otimes \bm H_B)$ such that $\mathcal E (\rho) = {\rm tr}_C [\rho]$ for $\rho \in L(\bm H_A \otimes \bm H_B \otimes \bm H_C)$.
$\Box$\end{proof}

Historically, the strong subadditivity (\ref{strong_subadditivity}) of the von Neumann entropy was first proved based on the Lieb theorem~\cite{Lieb1,Lieb2,Lieb3}.  Later,  the  monotonicity inequalities (\ref{monotonicity}) and (\ref{monotonicity2}) for the quantum relative entropy  were proved from the strong subadditivity~(\ref{strong_subadditivity})~\cite{Lindblad2,Lindblad3,Uhlmann}. 
On the other hand,  Petz~\cite{Petz2} proved  monotonicity~(\ref{monotonicity2}) without invoking the strong subadditivity.  Nielsen and Petz~\cite{Nielsen_Petz}  pedagogically discussed this proof.
 In a similar manner, Petz~\cite{Petz3} showed a direct proof of monotonicity~(\ref{monotonicity}), which we will discuss in Appendix~B.

\begin{remark}
In contrast to the quantum case, it is easy to prove the monotonicity of the classical relative entropy~\cite{Cover-Thomas}.  Let $\bm p := ( p(a) )_{a \in A}$ and $\bm q := (q(a))_{a \in A}$ be probability distributions on $A$.  
The classical counterpart of a quantum CPTP map is a Markov maps, which is given by transition probabilities $\{ r(b|a) \} _{a \in A, b \in B}$ with $\sum_b r(b|a) = 1$ such that
\begin{equation}
p'(b) := \sum_a r(b|a) p(a), \  q'(b) := \sum_a r(b|a) q(a).
\end{equation}
We write $\mathcal E (\bm p) := (p'(b))_{b \in B}$ and $\mathcal E (\bm q) := (q'(b))_{b \in B}$. Our goal is to show
\begin{equation}
S( \bm p \| \bm q) \geq S(\mathcal E(\bm p) \| \mathcal E (\bm q))
\label{monotonicity_classical}
\end{equation}
for the classical relative entropy.  Let $p(a,b) := r(b|a)p(a)$ and $q(a,b) := r(b|a)q(a)$.  We then have
\begin{equation}
\begin{split}
S( \bm p \| \bm q) &= \sum_{a,b} p(a,b) \ln \frac{p(a)}{q(a)} = \sum_{a,b} p(a,b) \ln \frac{p(a,b)}{q(a,b)} \\
&= S(\mathcal E(\bm p)  \| \mathcal E (\bm q)) + \sum_{a,b} p(a,b) \ln \frac{p(a|b)}{q(a|b)},
\end{split} 
\end{equation}
where $p(a|b) := p(a,b) / p'(b)$ and $q(a|b) := q(a,b) / q'(b)$.  By noting that \\ $ \sum_{a,b} p(a,b) \ln ( p(a|b) / q(a|b) ) \geq 0$ holds from the positivity of the  classical relative entropy, we obtain inequality~(\ref{monotonicity_classical}).  We note that the strong subadditivity of the Shannon entropy straightforwardly follows from the monotonicity of the classical relative entropy.
\end{remark}


\section{Quantum Mutual Information and Related Quantities}

In this section, we discuss the basic properties of the quantum mutual information and related quantities.
In particular, we introduce two important quantities that are closely related to the quantum mutual information:  the Holevo $\chi$-quantity and the QC-mutual information (the Groenewold-Ozawa information).  We discuss their information-theoretic meanings.

\subsection{Quantum Mutual Information}

We first introduce the quantum mutual information.

\begin{definition}
Let $\rho^{AB} \in Q(\bm H_A \otimes \bm H_B)$ be a quantum state and $\rho^A := {\rm tr}_B [\rho^{AB}]$ and $\rho^B := {\rm tr}_A [\rho^{AB}]$ be its partial states.  The mutual information between two systems  is then defined as
\begin{equation}
I^{A:B} (\rho^{AB}) := S(\rho^{AB} \| \rho^A \otimes \rho^B) = S(\rho^A) + S(\rho^B) - S(\rho^{AB}).
\end{equation}
\end{definition}

From the positivity of the quantum relative entropy, 
\begin{equation}
I^{A:B}(\rho^{AB}) \geq 0
\end{equation}
holds, where the equality is achieved if and only if $\rho^{AB} = \rho^A \otimes \rho^B$.

\begin{remark}
We note the relationship between the quantum mutual information and the classical mutual information.
Let $\{ | \varphi_a \rangle \}_{a \in A}$ and $\{ | \psi_b \rangle \}_{b \in B}$ be orthonormal bases of $\bm H_A$ and $\bm H_B$, respectively.  We define
\begin{equation}
\rho^{AB} := \sum_{ab} p(a,b) | \varphi_a \rangle \langle \varphi_a | \otimes | \psi_b \rangle \langle \psi_b |,
\end{equation}
where $\{ p(a,b) \}_{(a,b) \in A \times B}$ is a classical probability distribution on $A \times B$.  Then the quantum mutual information reduces to
\begin{equation}
I^{A:B}(\rho^{AB}) = \sum_{a,b} p(a,b) \ln \frac{p(a,b)}{p(a)p(b)},
\label{classical_mutual}
\end{equation}
where $p(a) := \sum_b p(a,b)$ and $p(b) := \sum_ap(a,b)$. The rhs of Eq.~(\ref{classical_mutual}) is the classical mutual information between $A$ and $B$.

\end{remark}

We now discuss the data processing inequality, which is a straightforward consequence of  monotonicity~(\ref{monotonicity}) of the quantum relative entropy.

\begin{theorem}[Data processing inequality]
Let $\mathcal E_A: L(\bm H_A) \to L(\bm H_{A'})$ and $\mathcal E_B : L(\bm H_B) \to L(\bm H_{B'})$ be CPTP maps.   The quantum mutual information is non-increasing by $\mathcal E_A \otimes \mathcal E_B$:
\begin{equation}
I^{A':B'}((\mathcal E_A \otimes \mathcal E_B) (\rho^{AB})) \leq I^{A:B}(\rho^{AB}).
\label{inequality1}
\end{equation}
\end{theorem}

\begin{proof}
Noting that $(\mathcal E_A \otimes \mathcal E_B) (\rho^A \otimes \rho^B) = \mathcal E_A (\rho^A) \otimes \mathcal E_B (\rho^B) $ and
\begin{equation}
{\rm tr}_{B'}[(\mathcal E_A \otimes \mathcal E_B) (\rho^{AB})] = \mathcal E_A (\rho^A), \ {\rm tr}_{A'}[(\mathcal E_A \otimes \mathcal E_B) (\rho^{AB})] = \mathcal E_B (\rho^B),
\end{equation}
inequality~(\ref{inequality1}) follows from the monotonicity~(\ref{monotonicity}) of the quantum relative entropy. 
$\Box$\end{proof}

The data processing inequality states that the quantum mutual information never increases by any CPTP map that is performed on each systems individually.  The following corollary is a special case.

\begin{corollary}
We consider three systems corresponding to $\bm H_A$, $\bm H_B$, and $\bm H_C$.  Then
\begin{equation}
I^{A:B} (\rho^{AB}) \leq  I^{A:BC} (\rho^{ABC}).
\label{inequality2}
\end{equation}
\end{corollary}

\begin{proof}
By taking $\mathcal E_{BC} (\rho^{BC}) := {\rm tr}_C [\rho^{BC}]$ and applying Theorem~4.1 to $\mathcal I_A \otimes  \mathcal E_{BC}$, we obtain inequality~(\ref{inequality2}).   
$\Box$\end{proof}

\begin{remark}
We note that inequality~(\ref{inequality2}) can be written as
\begin{equation}
S(\rho^A) + S(\rho^B) - S(\rho^{AB}) \leq S(\rho^A) + S(\rho^{BC}) - S(\rho^{ABC}),
\end{equation}
which is equivalent to  the strong subadditivity~(\ref{strong_subadditivity}).
\end{remark}


\subsection{Holevo's $\chi$-quantity}

We next introduce the Holevo's $\chi$-quantity (or just the $\chi$-quantity) that is related to the accessible classical information encoded in quantum states~\cite{Holevo,Yuen_Ozawa,Fuchs_Caves}.

\begin{definition}
Let $A$ be a finite set, $( p(a) )_{a \in A}$ be a probability distribution with $\sum_a  p(a) = 1$, and  $\rho_a \in Q(\bm H_S)$ be a quantum state labeled by $a \in A$.  The $\chi$-quantity is defined as
\begin{equation}
\chi^{AS} := S(\rho) - \sum_a p(a) S(\rho_a),
\label{Holevo_chi}
\end{equation}
where $\rho := \sum_a p(a) \rho_a$.
\end{definition}

We introduce an auxiliary system $\bm H_A$ with orthonormal basis $\{ | \varphi_a \rangle \}_{a \in A}$ that can store the classical information about $a$. We define a density operator 
\begin{equation}
\rho^{AS} := \sum_a p(a) | \varphi_a \rangle \langle \varphi_a | \otimes \rho_a.
\end{equation}
The $\chi$-quantity is then given by the mutual information
\begin{equation}
\chi^{AS} = I^{A:S} (\rho^{AS}),
\label{mutual_chi}
\end{equation}
which is a useful formula.

\begin{remark}
We note the relationship between $\chi$-quantity and the classical mutual information.
Let $\{ | \psi_b \rangle \}_{b \in B}$ be an orthonormal basis of $\bm H_S$.  We assume that $\rho_a$'s can simultaneously be diagonalized as
\begin{equation}
\rho_a = \sum_b p(b|a) | \psi_b \rangle \langle \psi_b |,
\end{equation}
where $\sum_b p(b|a) = 1$ for any $a$.  In this case, we can straightforwardly show that
\begin{equation}
\chi^{AS} = I^{A:B},
\end{equation}
where $I^{A:B}$ is the classical mutual information between $A$ and $B$ for the joint distribution $p(a,b) := p(b|a)p(a)$.
\end{remark}

The following theorems describe important properties of the $\chi$-quantity and the von Neumann entropy.

\begin{theorem}[Concavity of the von Neumann entropy]
The $\chi$-quantity satisfies
\begin{equation}
\chi^{AS} \geq 0,
\label{chi_inequality1}
\end{equation}
or equivalently
\begin{equation}
S(\rho) \geq \sum_a p(a) S(\rho_a),
\label{concavity}
\end{equation}
which is called the concavity of the von Neumann entropy.  The equality is achieved if and only if $p(a) = 1$ for a single $a$.
\end{theorem}

\begin{proof}
The inequality is obvious from the positivity of the mutual information in Eq.~(\ref{mutual_chi}).  The equality is achieved if and only if $\rho^{AS}$ is a product state, which implies that $p(a) = 1$ holds for a single $a$.
$\Box$\end{proof}

\begin{theorem}
The $\chi$-quantity satisfies
\begin{equation}
\chi^{AS} \leq H(\bm p),
\label{chi_inequality2}
\end{equation}
or equivalently
\begin{equation}
S(\rho) \leq H(\bm p) + \sum_a p(a) S(\rho_a),
\label{decomposition_inequality}
\end{equation} 
where $H(\bm p) := -\sum_a p(a) \ln p(a)$.  The equality is achieved if the supports of $\rho_a$'s are mutually orthogonal.
\end{theorem}

\begin{proof}
We introduce an auxiliary system $\bm H_{A'}$ with orthonormal basis $\{ | \psi_a \rangle \}_{a \in A}$.  We define a state
\begin{equation}
\sigma^{AA'} := \sum_a p(a) | \varphi_a \rangle \langle \varphi_a | \otimes | \psi_a \rangle \langle \psi_a | \in Q(\bm H_A \otimes \bm H_{A'}),
\end{equation}
where the mutual information between $A$ and $A'$  is given by $I^{A:A'}(\sigma^{AA'}) = H(\bm p)$.  
On the other hand, we define a CPTP map $\mathcal E : L(\bm H_{A'}) \to L(\bm H_S)$ such that
\begin{equation}
\mathcal E (| \psi_a \rangle \langle \psi_a |) = \rho_a
\label{rho_a}
\end{equation} 
for any $a$.  In fact, we can construct $\mathcal E$ satisfying Eq.~(\ref{rho_a}) as follows.  Let $\rho_a = \sum_i q_a(i) | a i \rangle \langle ai |$ be the spectrum decomposition of $\rho_a$.  We define Kraus operators
\begin{equation}
M_{ai} := \sqrt{q_a(i)} | a i \rangle \langle \psi_a | \in L(\bm H_{A'}, \bm H_S),
\label{Kraus3}
\end{equation}
which satisfies
\begin{equation}
\sum_{ai} M_{ai}^\dagger M_{ai} = \sum_{ai} q_a(i) | \psi_a \rangle \langle ai | a i \rangle \langle \psi_a | = \sum_{a}   | \psi_a \rangle \langle \psi_a | = I^{A'},
\end{equation}
where $I^{A'}$ is the identity on $\bm H_{A'}$.
By defining $\mathcal E$ with the Kraus operators (\ref{Kraus3}), we have
\begin{equation}
\mathcal E (| \psi_a \rangle \langle \psi_a |)  = \sum_i q_a(i) |ai \rangle \langle ai | = \rho_a,
\end{equation}
which confirms Eq.~(\ref{rho_a}).

By applying $\mathcal E \otimes \mathcal I_A$ to $\sigma^{AA'}$ with $\mathcal I_A$ the identity on $L(\bm H_A)$, we have
\begin{equation}
 (\mathcal E \otimes \mathcal I_A) (\sigma^{AA'}) = \sum_a p(a) \rho_a \otimes | \varphi_a \rangle \langle \varphi_a | =: \rho^{AS}. 
\end{equation}
Therefore, from the data processing inequality (\ref{inequality1}), we obtain
\begin{equation}
H(\bm p) = I^{A:A'} (\sigma^{AA'}) \geq  I^{A:S} (\rho^{AS}) = \chi^{AS},
\end{equation}
which implies inequality~(\ref{decomposition_inequality}).  If the supports of $\rho_a$'s are mutually orthonormal, the equality in (\ref{decomposition_inequality}) is achieved because of Eq.~(\ref{decomposition}).
$\Box$\end{proof}

Theorem~4.3 implies that Eq.~(\ref{decomposition}) is replaced by inequality~(\ref{decomposition_inequality}) if the supports of $\rho_a$'s are not  mutually orthogonal.
The following corollary is a direct consequence of Theorem~4.3.

\begin{corollary}
We define $\rho := \sum_a p(a) | \phi_a \rangle \langle \phi_a|$ with $\sum_a p(a) = 1$, where $| \phi_a \rangle$'s are not necessarily mutually-orthogonal.   We then have
\begin{equation}
S(\rho) \leq H(\bm p),
\end{equation}
where $H(\bm p) := -\sum_a p(a) \ln p(a)$.
\end{corollary}

\begin{proof}
Apply Theorem~4.3 to  $\rho_a := | \phi_a \rangle \langle \phi_a|$.
$\Box$\end{proof}

We next show the data processing inequality for the $\chi$-quantity.

\begin{theorem}[Data processing inequality]
We define
\begin{equation}
\chi'^{AS'} := S (\mathcal E (\rho)) - \sum_a p(a) S(\mathcal E (\rho_a)),
\label{inequality3}
\end{equation}
where  $\mathcal E: L(\bm H_S) \to L(\bm H_{S'})$ is a CPTP map. 
Then
\begin{equation}
\chi'^{AS'}  \leq \chi^{AS}.
\end{equation}
\end{theorem}

\begin{proof}
From Eq.~(\ref{mutual_chi}) and  inequality (\ref{inequality1}), we have $\chi'^{AS'} = I^{A:S'} ((\mathcal I_A \otimes \mathcal E  ) (\rho^{AS})) \leq I^{A:S} (\rho^{AS}) = \chi^{AS}$, where $\mathcal I_A$ is the identity on $L(\bm H_A)$.
$\Box$\end{proof}

We next formulate and prove the Holevo bound, which determines the upper bound of the accessible classical information that is encoded in a quantum system.  
We consider that the classical information about $a \in A$ is encoded in a quantum state $\rho_a \in Q(\bm H_S)$.
We extract the information about $a$ by performing a quantum measurement on the quantum system.

Let $\{ E_b \}_{b \in B}$ be a POVM with a finite set $B$.  The probability of obtaining outcome $b \in B$ by a  measurement on state $\rho_a$ with the POVM is given by 
\begin{equation}
p(b|a) = {\rm tr} [E_b \rho_a].
\end{equation}  
The joint probability of $(a,b)$ is  $p(a,b) = p(b|a)p(a)$, whose marginal distributions are $p(a) := \sum_b p(a,b)$ and $p(b) := \sum_a p(a,b)$.  The mutual information $I^{A:B}$ between the two classical variables is then given by
\begin{equation}
I^{A:B} = \sum_{a,b} p(a,b) \ln \frac{p(a,b)}{p(a)p(b)}.
\label{Holevo_bound_classical}
\end{equation}  
The Holevo bound states that the upper bound of the classical mutual information~(\ref{Holevo_bound_classical}) is bounded by the $\chi$-quantity.

\begin{theorem}[Holevo bound]
\begin{equation}
I^{A:B} \leq \chi^{AS}
\label{Holevo}
\end{equation}
 holds for any POVM $\{ E_b \}_{b \in B}$.
\end{theorem}

We will discuss two proofs of Theorem~4.5 in the followings.  The first proof~\cite{Nielsen-Chuang} is more intuitive than the second one~\cite{Yuen_Ozawa}, while the second one is simpler than the first one.  Both proofs are based on the monotonicity of the quantum relative entropy.

\begin{proof}
We now discuss the first proof.
We introduce two auxiliary quantum systems described by Hilbert spaces $\bm H_A$ and $\bm H_B$.
Their orthonormal bases are  labeled by the corresponding  classical variables as $\{ | a \rangle \}_{a \in A} \subset \bm H_A$ and $\{ | b \rangle \}_{b \in B} \subset \bm H_B$.  We define $\rho^{AS} \in Q(\bm H_A \otimes \bm H_S)$ as
\begin{equation}
\rho^{AS} := \sum_{a} p(a) | a \rangle \langle a | \otimes \rho_a,
\end{equation}
and define $\rho^{ASB} := \rho^{AS} \otimes | 0 \rangle \langle 0 | \in Q(\bm H_A \otimes \bm H_S \otimes \bm H_B)$, where $| 0 \rangle \in \bm H_B$ is an initial reference state.  It is easy to show that there is a CPTP map $\mathcal E_{SB}$ acting on $L(\bm H_S \otimes \bm H_B)$ such that 
\begin{equation}
(\mathcal I_A \otimes \mathcal E_{SB}) (\rho^{ASB}) = \sum_{a,b} p(a) | a \rangle \langle a | \otimes \sqrt{E_b} \rho_a \sqrt{E_b} \otimes |b \rangle \langle b | =: \rho'^{ASB},
\end{equation}
where $\mathcal I_A$ is the identity on $L(\bm H_A)$. We note that  $\mathcal E_{SB}$ describes a measurement process corresponding to POVM $\{ E_a \}$.

 Let  $\rho'^{AB} := {\rm tr}_S [\rho^{ASB}]$.  We then obtain
\begin{equation}
\chi^{AS} = I^{A:S}(\rho^{AS}) =  I^{A:SB}(\rho^{ASB}) \geq I^{A:SB} (\rho'^{ASB}) \geq I^{A:B} (\rho'^{AB}),
\label{Holevo1}
\end{equation}
where we used Eq.~(\ref{mutual_chi}) and the data processing inequalities (\ref{inequality1}) and (\ref{inequality2}).  By noting that
\begin{equation}
\rho'^{AB} = \sum_{a,b} p(a,b) | a \rangle \langle a | \otimes | b \rangle \langle b |,
\end{equation}
we obtain 
\begin{equation}
I^{A:B} (\rho'^{AB}) = I^{A:B},
\label{Holevo2} 
\end{equation}
where the rhs means the classical mutual information~(\ref{Holevo_bound_classical}).  Inequality~(\ref{Holevo1}) and Eq.~(\ref{Holevo2}) imply the Holevo bound~(\ref{Holevo}).
$\Box$\end{proof}

\begin{proof}
We next discuss the second proof.
We note that the $\chi$-quantity can be written as
\begin{equation}
\chi^{AS} := \sum_a p(a) S(\rho_a \| \rho).
 \label{Holevo3}
\end{equation}
Let $N < \infty$ be the number of possible outcomes, and let $\rho_a'$, $\rho' \in Q(\mathbb C^N)$ be diagonal matrices, where every diagonal element is given by ${\rm tr}[E_b \rho_a]$ or  ${\rm tr}[E_b \rho]$,  respectively. We have
\begin{equation}
I^{A:B} = \sum_a p(a) S(\rho_a' \| \rho').
 \label{Holevo4}
\end{equation}
On the other hand, a linear map $\mathcal E: L(\bm H_S) \to L(\mathbb C^N)$ satisfying $\mathcal E(\rho_a) = \rho_a'$ and $\mathcal E(\rho) = \rho'$ is CPTP.  Therefore, from the monotonicity of the quantum relative entropy, we obtain
 \begin{equation}
 S(\rho_a \| \rho) \geq S(\rho_a' \| \rho')
 \label{Holevo5}
 \end{equation}
 for any $a$.  By averaging inequality~(\ref{Holevo5}) over all $a$'s and by using Eqs.~(\ref{Holevo3}) and (\ref{Holevo4}), we obtain the Holevo bound~(\ref{Holevo}).
$\Box$\end{proof}


\subsection{QC-mutual Information (Groenewold-Ozawa Information)}

We next introduce a quantity called the QC-mutual information that is also related to the accessible classical information encoded in quantum states~\cite{Groenewold,Ozawa1,Buscemi,Sagawa-Ueda1}.
We consider a quantum measurement described by  POVM $\{ E_b \}_{b \in B}$, where $B = \{ b \}$ is  the finite set of outcomes.  If the measured state is   $\rho \in Q(\bm H_S)$, the  probability of obtaining outcome $b$ is given by $p(b) = {\rm tr}_S[E_b \rho]$.  By defining  
\begin{equation}
\rho_b := \frac{1}{p(b)}\sqrt{E_b} \rho \sqrt{E_b},
\end{equation}
we introduce the QC-mutual information as follows.

\begin{definition}
In the above setup, the QC-mutual information (the Groenewold-Ozawa information) is defined as 
\begin{equation}
I_{\rm QC}^{S:B} := S(\rho) - \sum_b p(b) S(\rho_b).
\label{QC_mutual}
\end{equation}
\end{definition}

We note that the QC-mutual information~(\ref{QC_mutual})  only depends on the measured state $\rho$ and the POVM $\{ E_b \}_{b \in B}$.  We note that Groenewold~\cite{Groenewold} and Ozawa~\cite{Ozawa1} originally discussed the case that any  Kraus operator $M_b \in  L(\bm H_S)$ satisfies $E_b := M_b^\dagger M_b$.

\begin{remark}
We note the relationship between the QC-mutual information and the classical mutual information.
Let $\{ | \varphi_a \rangle \}_{a \in A}$ be an orthonormal basis of $\bm H_S$.  We assume that $\rho$ and $E_b$'s are simultaneously diagonalized such that $\rho = \sum_a p(a) | \varphi_a \rangle \langle \varphi_a |$ and $E_b = \sum_a p(b|a) | \varphi_a \rangle \langle \varphi_a |$ for any $b \in B$.  In this case, we have $\rho_b = \sum_a p(a|b) | \varphi_a \rangle \langle \varphi_a |$ with $p(a|b) := p(b|a)p(a) /  \left( \sum_a   p(b|a)p(a) \right)$.  Therefore,
we obtain
\begin{equation}
I_{\rm QC}^{S:B} = I^{A:B},
\end{equation}
where $I^{A:B}$ is the classical mutual information between $A$ and $B$ for the joint distribution $p(a,b) := p(b|a)p(a)$.
\end{remark}

We consider a quantum measurement with the set of Kraus operators $\{ M_b \}_{b \in B} \subset L(\bm H_S)$ satisfying $E_b = M_b^\dagger M_b$. Let $\rho'_b := M_b \rho M_b^\dagger / p(b)$ and $\rho' := \sum_b p(b) \rho'_b$. The QC-mutual information can then be written as  
\begin{equation}
I_{\rm QC}^{S:B} = \chi^{SB} - \Delta S_{\rm meas},
\end{equation} 
where $\chi^{SB} := S(\rho') - \sum_b p(b) S(\rho_b')$ is the $\chi$-quantity of the post-measurement states, and $\Delta S_{\rm meas} := S(\rho') - S(\rho)$ is the change of  the von Neumann entropy by the measurement.

On the other hand, the QC-mutual information $I_{\rm QC}^{S:B}$ equals the $\chi$-quantity of an auxiliary system.
Let  $\bm H_R$ be the Hilbert space of the auxiliary system $R$, and  $| \Psi \rangle \in \bm H_S \otimes \bm H_R$ be a purification of $\rho$ such that
\begin{equation}
{\rm tr}_R [|\Psi \rangle \langle \Psi |] = \rho.
\end{equation}
We define $\rho^R := {\rm tr}_S [ |\Psi \rangle \langle \Psi | ]$ and 
\begin{equation}
\rho^R_b := {\rm tr}_S [ (\sqrt{E_b} \otimes I^R) | \Psi \rangle \langle \Psi | (\sqrt{E_b} \otimes I^R)  ] / p(b),  
\end{equation}
where $I^R \in L(\bm H_R)$ is the identity.   We note that $\rho^R = \sum_b p(b) \rho^R_b$.  We then obtain the following theorem:

\begin{theorem} The QC-mutual information satisfies
\begin{equation}
I_{\rm QC}^{S:B} = \chi^{BR},
\label{QC_chi1}
\end{equation}
where $\chi^{BR} := S(\rho^R) - \sum_b p(b) \rho^R_b$ is the $\chi$-quantity of $\{ \rho^R_b \}_{b \in B}$.
\end{theorem}

\begin{proof}
By noting that ${\rm tr}_R [ (\sqrt{E_b} \otimes I^R) | \Psi \rangle \langle \Psi | (\sqrt{E_b} \otimes I^R)   ] / p(b) = \rho_b$,  we  have $S(\rho_b) = S(\rho^R_b)$ and $S(\rho) = S(\rho^R)$, which imply Eq.~(\ref{QC_chi1}).
 $\Box$\end{proof}
 
Therefore, the QC-mutual information satisfies the following inequality.

\begin{corollary}
\begin{equation}
0 \leq I_{\rm QC}^{S:B} \leq H (\bm p),
\label{QC_inequality}
\end{equation}
where $H(\bm p) := - \sum_b p(b) \ln p(b)$.
\end{corollary}

\begin{proof}
Apply inequalities~(\ref{chi_inequality1}) and (\ref{chi_inequality2}) to $\chi^{BR}$.
$\Box$\end{proof}

We next consider an information-theoretic meaning of the QC-mutual information.
We assume that classical information about $a \in A$ is encoded in $\rho \in Q(\bm H_S)$ as
\begin{equation}
\rho = \sum_a q(a) \rho_a,
\end{equation}
where $\rho_a$'s are density operators, $q(a)$'s satisfy $\sum_a q(a) = 1$, and $A$ is assumed to be a finite set.   We then perform a measurement with  POVM $\{ E_b \}_{b \in B}$.  The probability of obtaining $b$ under the condition of $a$ is given by
\begin{equation}
p(b|a) = {\rm tr}_S [E_b \rho_a].
\end{equation}
The joint distribution is  $p(a,b) := p(b|a) q(a)$.
We note that the unconditional probability of obtaining $b$ is given by $p(b) = \sum_a p(a,b) = {\rm tr}_S [E_a \rho]$, 
and the QC-mutual information is defined by Eq.~(\ref{QC_mutual}) with  $\rho_b := \sqrt{E_b} \rho \sqrt{E_b} / p(b)$.
We then have the following theorem as a ``dual'' of the Holevo bound (Theorem~4.5)~\cite{Buscemi}.

\begin{theorem}
In the above setup, the classical mutual information $I^{A:B}$ between $A$ and $B$ is bounded by the QC-mutual information:
\begin{equation}
I^{A:B} \leq I_{\rm QC}^{S:B}.
\label{Holevo_dual}
\end{equation}  
\end{theorem}

\begin{proof}
Let $\rho_a = \sum_i q_a(i) | \psi_{ai} \rangle \langle \psi_{ai}|$ be the spectrum decomposition of $\rho_a$, where $\{ | \psi_{ai} \rangle \}_{i}$ is an orthonormal basis of $\bm H_S$.
We introduce an auxiliary system described by $\bm H_R$ with orthonormal basis $\{ | r_{ai} \rangle \}_{ai}$, and define a purification of $\rho$: 
\begin{equation}
| \Psi \rangle := \sum_{ai} \sqrt{q(a)q_a(i)} | \psi_{ai} \rangle | r_{ai} \rangle \ \in \bm H_S \otimes \bm H_R.
\end{equation}
We have $\rho = {\rm tr}_R [ | \Psi \rangle \langle \Psi | ]$ and $\rho_a = {\rm tr}_R[I^S \otimes P^R_a | \Psi \rangle \langle \Psi |] / q(a)$, where $P^R_a := \sum_i | r_{ai} \rangle \langle r_{ai}|$ and $I^S$ is the identity on $\bm H_S$.
On the other hand, we define $\rho^R := {\rm tr}_S [| \Psi \rangle \langle \Psi |]$ and $\rho_b^R := {\rm tr}_S [(\sqrt{E_b} \otimes I^R)| \Psi \rangle \langle \Psi | (\sqrt{E_b} \otimes I^R) ] / p(b) =  {\rm tr}_S [(E_b \otimes I^R)| \Psi \rangle \langle \Psi | ] / p(b) $, where $I^R$ is the identity on $\bm H_R$.
By noting that
\begin{equation}
{\rm tr}_S \left[ E_b {\rm tr}_R \left[ I^S \otimes P_a^R | \Psi \rangle \langle \Psi | \right] \right] = {\rm tr}_R \left[ P^R_a {\rm tr}_S \left[ E_b \otimes I^R | \Psi \rangle \langle \Psi | \right] \right],
\end{equation}
we obtain 
\begin{equation}
{\rm tr}_S [E_b \rho_a] q(a) = {\rm tr}_R [P_a^R \rho_b^R] p(b),
\end{equation}
and therefore
\begin{equation}
p(a,b) = {\rm tr}_R [P^R_a \rho_b^R] p(b).
\end{equation}
Since $I_{\rm QC}^{S:B} = S(\rho^R) - \sum_b p(b) S(\rho^R_b)$ holds from Theorem~4.6, we obtain inequality~(\ref{Holevo_dual}) by applying the Holevo bound to $\{ \rho_b^R \}_{b \in B}$.
$\Box$\end{proof}

We note that, in the set up of the Holevo bound (Theorem~4.5), the encoding of the classical information is fixed and the measurement is arbitrary.  In contrast, in the setup of Theorem~4.7, the encoding is arbitrary and the measurement is fixed.  Inequality~(\ref{Holevo_dual}) determines the upper bound of the accessible classical information under the condition that the measurement is given by $\{ E_b \}_{b \in B}$ and the ensemble average of the encoded states is given by $\rho$.


\section{Second Law of Thermodynamics}

We now derive the second law of thermodynamics in three manners.  The first derivation is based on the positivity of the quantum relative entropy.  The second derivation is based on the quantum fluctuation theorem, which is shown to be equivalent to the first derivation.  The third one is based on the monotonicity of the quantum relative entropy.

\subsection{Thermodynamic Entropy and the von Neumann Entropy}

Before going to the main part of this section, we briefly discuss the relationship between the thermodynamic entropy and the von Neumann entropy.  We consider a thermodynamic system that is in thermal equilibrium.
Let $H$ be the Hamiltonian of the system.  The Helmholtz free energy is  defined as
\begin{equation}
F := -\beta^{-1} \ln {\rm tr} [e^{-\beta H}],
\end{equation}
where  $\beta > 0$ is the inverse temperature of the system.
The thermodynamic entropy $S_{\rm therm}$ satisfies
\begin{equation}
S_{\rm therm} = \beta (\langle E \rangle - F),
\label{entropy_free_energy}
\end{equation}
where $\langle E \rangle$ is the average energy of the system.  
The thermodynamic relation (\ref{entropy_free_energy}) has been well established from the 19th century as a phenomenological thermodynamic relation for macroscopic systems.
In terms of statistical mechanics, however, it is not so obvious to determine the microscopic expression of  the thermodynamic entropy $S_{\rm therm}$ as will be discussed in Sec.~7.

As a special case, if we select the canonical distribution  
\begin{equation}
\rho_{\rm can} := e^{\beta (F - H)}
\end{equation}
as a microscopic expression of the thermal equilibrium state, we can easily show that the thermodynamic entropy is given by the von Neumann entropy of the system:

\begin{theorem}
The von Neumann entropy of $\rho_{\rm can}$ satisfies
\begin{equation}
S(\rho_{\rm can}) = \beta (\langle E \rangle_{\rm can} - F),
\label{Neumann_free_energy}
\end{equation}
where $\langle E \rangle_{\rm can} := {\rm tr} \left[ H \rho_{\rm can} \right]$ is the average energy in the canonical distribution.
\end{theorem}

The statistical-mechanical relation~(\ref{Neumann_free_energy}) is consistent with the thermodynamic relation~(\ref{entropy_free_energy}) with the correspondence between $S(\rho_{\rm can})$ and $S_{\rm therm}$.  Therefore, in the main part of this article, we will identify a canonical distribution to a thermal equilibrium state, and the von Neumann entropy to the thermodynamic entropy; these identifications  have been widely used in statistical mechanics.  Some subtle points on the validities of these identifications will be discussed in Sec.~7.

We can also easily calculate the quantum relative entropy of any state $\rho \in Q(\bm H)$ to the canonical distribution $\rho_{\rm can}$ as
\begin{equation}
S(\rho \| \rho_{\rm can}) = \beta ( F - \langle E \rangle ) - S(\rho), 
\end{equation}
where $\langle E \rangle := {\rm tr}[H \rho]$ is the average energy of $\rho$.  From the positivity of the quantum relative entropy, we have
\begin{equation}
 S(\rho) \leq \beta ( F - \langle E \rangle ),
 \label{Neumann_free_energy2}
\end{equation}
where the equality is achieved only if $\rho = \rho_{\rm can}$.  Inequality~(\ref{Neumann_free_energy2}) implies that  the  von Neumann entropy takes the maximum in the canonical distribution among the states that have the same amount of the energy.


\subsection{From the Positivity of the Quantum Relative Entropy}

We now derive the second law of thermodynamics on the basis of the positivity of the quantum relative entropy and the unitarity of the time evolution of the system~\cite{Tasaki1,Esposito4}.  We first prove a very general but almost trivial equality and inequality, and next apply them to thermodynamic situations.  Therefore, the nontrivial part of this subsection is in the applications to each examples.

We consider a unitary evolution of the system from state $\rho_{\rm i} \in Q(\bm H)$ to $\rho_{\rm f} \in Q(\bm H)$ such that $\rho_{\rm f} = U \rho_{\rm i} U^\dagger$.  We also introduce a reference state $\rho_0 \in Q(\bm H)$, which is different from $\rho_{\rm i}$ or $\rho_{\rm f}$ in general.  Since $S(\rho_{\rm i}) = S(\rho_{\rm f})$ holds due to the unitary evolution,  we have the following theorem:

\begin{theorem}
In the above setup,
\begin{equation}
- {\rm tr} [\rho_{\rm f} \ln \rho_0] - S(\rho_{\rm i}) = S(\rho_{\rm f} \| \rho_0)
\label{second_positive0}
\end{equation}
holds, which leads to 
\begin{equation}
- {\rm tr} [\rho_{\rm f} \ln \rho_0] - S(\rho_{\rm i}) \geq 0.
\label{second_positive}
\end{equation}
\end{theorem}

We note that we use the positivity of the relative entropy $S(\rho_{\rm f} \| \rho_0)$ to derive inequality~(\ref{second_positive}).  While Eq.~(\ref{second_positive0}) and inequality~(\ref{second_positive}) are obvious, they play key roles to derive the second law of thermodynamics as shown below.  We discuss two typical situations in the following.

\begin{example}
We assume that the system is driven by a time-dependent Hamiltonian $H(t)$ from $t=0$ to $t = \tau$, which gives the unitary operator as
\begin{equation}
U = {\rm T}\exp \left( - {\rm i} \int_0^\tau H(t) dt \right).
\label{unitary}
\end{equation}
The free energy corresponding to the Hamiltonian at time $t$ is given by
\begin{equation}
F(t) := -\beta^{-1} \ln {\rm tr} \left[ e^{-\beta H(t)}  \right],
\end{equation}
where $\beta > 0$.
We assume that the system is initially in the canonical distribution at inverse temperature $\beta$ such that 
\begin{equation}
\rho_{\rm i} := e^{\beta ( F(0) - H (0) )},
\label{positive1}
\end{equation}
and define the reference state as
\begin{equation}
\rho_0 :=  e^{\beta ( F(\tau) - H (\tau) )}.
\label{positive2}
\end{equation}
We stress that $\rho_{\rm f}  \neq \rho_0$ in general.  In this setup, we have
\begin{equation}
{\rm tr} \left[ \rho_{\rm f} \ln \rho_0 \right] = \beta \left( F(\tau) - {\rm tr} \left[ H \rho_{\rm f} \right] \right).
\end{equation}
Therefore, Eq.~(\ref{second_positive0}) reduces to
\begin{equation}
\beta (\langle W \rangle - \Delta F ) = S(\rho_{\rm f} \| \rho_0),
\label{second_positive2}
\end{equation}
where
\begin{equation}
\langle W \rangle :=  {\rm tr} [ \rho_{\rm f} H(\tau)] - {\rm tr} [ \rho_{\rm i} H(0)]
\label{work1}
\end{equation}
is the energy difference of the system, and
\begin{equation}
\Delta F := F(\tau) - F(0)
\end{equation}
is the free-energy difference corresponding to the initial and final Hamiltonians. We note that the energy difference $\langle W \rangle$ is regarded as the work performed on the system in this setup, because any heat bath is not attached to the system.  
Corresponding to inequality (\ref{second_positive}), we obtain the second law of thermodynamics
\begin{equation}
\langle W \rangle \geq \Delta F.
\label{second_positive3}
\end{equation}

\end{example}

\begin{example}
We assume that the total system consists of the main system $S$ and heat baths $B_k$ ($k=1,2, \cdots$).   The Hilbert spaces corresponding to $S$ and $B_k$ are respectively  given by $\bm H_S$ and $\bm H_{B_k}$ so that $\bm H = \bm H_S \otimes_k \bm H_{B_k}$. 
Let $H^S \in L(\bm H_S)$ be the system's Hamiltonian, $H^{B_k} \in L(\bm H_{B_k})$ be the $k$th Bath's Hamiltonian, and $H^{SB_k} \in L(\bm H_S \otimes \bm H_{B_k})$ be the interaction Hamiltonian between $S$ and $B_k$.  We assume that $H^S$ and $H^{SB_k}$ are time-dependent, while $H^{B_k}$ is time-independent. 
The total system then obeys a unitary evolution from $t=0$ to $t= \tau$ corresponding to the total Hamiltonian
\begin{equation}
H(t) = H^S(t)  + \sum_k ( H^{SB_k}(t)  +   H^{B_k} ),
\label{total_Hamiltonian}
\end{equation}
where we omitted to write the tensor products with the  identities on $\bm H_S$ and $\bm H_{B_k}$'s.  
Let   $H^{\rm int}(t) := \sum_k H^{SB_k}(t)$. For simplicity, we assume that the interaction Hamiltonian satisfies  $H^{\rm int} (0) = H^{\rm int}(\tau) = 0$.  

Let $\rho_{\rm can}^{B_k}$ be the canonical distribution corresponding to $H^{B_k}$ such that
\begin{equation}
\rho_{\rm can}^{B_k} := e^{\beta_k (F_k -H^{B_k} ) },
\end{equation}
where $\beta_k > 0$ is the inverse temperature of $B_k$, and  
\begin{equation}
F_k := - \beta_k^{-1} \ln {\rm tr} \left[ e^{-\beta H^{B_k}} \right].
\end{equation}
We assume that the initial state of the total system is given by a product state
\begin{equation}
\rho_{\rm i} := \rho_{\rm i}^S \otimes _k \rho_{\rm can}^{B_k},
\label{initial1}
\end{equation}
where $\rho_{\rm i}^S$ is an arbitrary initial state of the system.  We note that Eq.~(\ref{initial1}) is consistent with assumption $H^{\rm int} (0) = 0$.  The final state is given by $\rho_{\rm f} = U \rho_{\rm i} U^\dagger$, where $U$ is given by Eq.~(\ref{unitary}) with the total Hamiltonian (\ref{total_Hamiltonian}).  The final state of $S$ is given by
\begin{equation}
\rho_{\rm f}^S := {\rm tr}_{B} [\rho_{\rm f}],
\end{equation}
where ${\rm tr}_{\rm B}$ means the trace over all $\bm H_{B_k}$'s.  We then define the reference state as
\begin{equation}
\rho_0 := \rho_{\rm f}^S \otimes_k \rho_{\rm can}^{B_k}.
\label{ref1}
\end{equation}
In this setup, we can show that Eq.~(\ref{second_positive0}) is equivalent to
\begin{equation}
\Delta S - \sum_k \beta_k \langle Q_k \rangle = S(\rho_{\rm f} \| \rho_0),
\label{second_positive4}
\end{equation}
where 
\begin{equation}
\Delta S := S(\rho_{\rm f}^S) - S(\rho_{\rm i}^S)
\end{equation}
is the difference in the von Neumann entropy of the system, and
\begin{equation}
\langle Q_k \rangle := {\rm tr}\left[ H^{B_k} \rho_{\rm i} \right] - {\rm tr} \left[ H^{B_k}\rho_{\rm f} \right]
\label{heat1}
\end{equation}
is  regarded as the heat that is absorbed by $S$ from bath $B_k$ due to assumption  $H^{\rm int} (0) = H^{\rm int}(\tau) = 0$.  Corresponding to inequality~(\ref{second_positive}), we obtain the Clausius inequality
\begin{equation}
\Delta S - \sum_k \beta_k \langle Q_k \rangle \geq 0.
\label{second_positive5}
\end{equation}
We note that, in the conventional thermodynamics, the initial and final states of the system are assumed to be in thermal equilibrium.  On the other hand, we assumed nothing on $\rho_{\rm i}^S$ and $\rho_{\rm f}^S$ above.  Therefore, inequality~(\ref{second_positive5}) is regarded as a generalization of the conventional Clausius inequality to situations in which the initial and final states of the system are out of equilibrium.

In the following, we additionally assume that the initial state of the system is given by the canonical distribution at inverse temperature $\beta$ such that
\begin{equation}
\rho_{\rm i}^S = e^{\beta ( F^S(0) - H^S (0) )},
\end{equation}
where 
\begin{equation}
F^S(t) := -\beta^{-1} {\rm tr} \left[ e^{-\beta H^S(t)}  \right].
\end{equation}
This assumption is consistent with assumption $H^{\rm int} (0)  = 0$.   By using notation $\rho_0^S := e^{\beta ( F^S(\tau) - H^S (\tau) )}$, we obtain
\begin{equation}
S(\rho_{\rm f}^S) \leq - {\rm tr} \left[ \rho_{\rm f}^S \ln \rho_0^S \right] = \beta \left( F(\tau) - {\rm tr} \left[ H^S \rho_{\rm f}^S \right] \right),
\end{equation}
where we used the positivity of $S(\rho_{\rm f}^S \| \rho_0^S)$.  Therefore, we obtain
\begin{equation}
\langle \Delta E^S \rangle - \Delta F^S   \geq \Delta S,
\label{free_energy_entropy}
\end{equation}
where
\begin{equation}
\langle \Delta E^S \rangle := {\rm tr} [ \rho_{\rm f}^S H^S(\tau)] - {\rm tr} [ \rho_{\rm i}^S H^S(0)]
\end{equation}
is the energy difference of the system.
By combining inequalities~(\ref{second_positive5}) and (\ref{free_energy_entropy}), we obtain
\begin{equation}
\beta \left( \langle \Delta E^S \rangle - \Delta F^S \right) - \sum_k \beta_k \langle Q_k \rangle \geq 0.
\label{second_positive6}
\end{equation}
For a special  case in which there is a single heat bath at inverse temperature $\beta$, inequality~(\ref{second_positive}) reduces to
\begin{equation}
\langle W \rangle \geq \Delta F^S,
\label{second_positive7}
\end{equation}
where
\begin{equation}
\langle W \rangle := \langle E^S \rangle - \langle Q \rangle
\label{first_law}
\end{equation}
is the work performed on the system.

\end{example}

The argument in this subsection is based on the positivity of the quantum relative entropy and the unitary evolution of the total system.  We can replace the unitary evolution by  a unital CPTP map $\mathcal E$ satisfying $\mathcal E (I) = I$.  In this case, the von Neumann entropy of the total system is non-decreasing as $S(\rho_{\rm i}) \leq S(\rho_{\rm f})$, which has been shown in Corollary~3.2.  Thus, Eq.~(\ref{second_positive0}) is replaced by an inequality
\begin{equation}
-{\rm tr}[\rho_{\rm f} \ln \rho_0] - S(\rho_{\rm i}) \geq S(\rho_{\rm f} \| \rho_0),
\end{equation}
and therefore, inequality~(\ref{second_positive}) remains unchanged.   As a consequence, inequalities~(\ref{second_positive3}) and (\ref{second_positive5}) still hold for such a CPTP map $\mathcal E$ acting on the total system.


\subsection{From the Quantum Fluctuation Theorem}

The quantum fluctuation theorem gives  information about fluctuations of the entropy production~\cite{Tasaki1,Yukawa,Mukamel,Jarzynski3,Roeck,Monnai,Esposito1,Talkner1,Talkner2,Esposito2,Saito,Gaspard,Huber,Utsumi1,Utsumi2,Esposito3,Andireux,Hanggi1,Hanggi2,Nakamura,Ohzeki,Campisi,Lutz,Horowitz1,Horowitz4}. 
We first introduce the stochastic entropy production and formulate the quantum fluctuation theorem in a very general setup.  
Let $\rho_{\rm i}, \rho_0 \in Q(\bm H)$ be density operators.  They have  spectrum decompositions  
 $\rho_{\rm i} = \sum_a p_{\rm i}(a) | \psi_a \rangle \langle \psi_a |$ and $\rho_0 = \sum_b p_0 (b) | \phi_b \rangle \langle \phi_b |$, where $\{ | \psi_a \rangle \}$ and $\{ | \phi_b \rangle \}$ are orthonormal basis of $\bm H$.  
To formulate the quantum fluctuation theorem,  the key concepts are the forward and backward processes that are described as follows.

\

\textit{Forward process.}
In the forward process, the initial state is given by $\rho_{\rm i}$.  We first perform the projection measurement on $\rho_{\rm i}$ with basis $\{ | \psi_a \rangle \}$, and obtain outcome $a$ with probability $p_{\rm i}(a)$. By this measurement, the ensemble average of the post-measurement states equals $\rho_{\rm i}$.
We next perform  a unitary operation with a time-dependent Hamiltonian $H(t)$ from $t=0$ to $\tau$.  The unitary operator is given by Eq.~(\ref{unitary}).  The density operator of the system then becomes $\rho_{\rm f} = U \rho_{\rm i} U^\dagger$. 
We next perform the projection measurement on $\rho_{\rm f}$ with basis $\{ | \phi_b \rangle \}$, and obtain outcome $b$ with probability $p_{\rm f} (b) := \langle \phi_b | \rho_{\rm f} | \phi_b \rangle$.  The joint probability of $(a,b)$ is given by
\begin{equation}
p(a,b) := p(b \leftarrow a) p_{\rm i}(a),
\end{equation}
where
\begin{equation}
p(b \leftarrow a) := | \langle \phi_b | U | \psi_a \rangle |^2
 \end{equation}
 is the transition probability.  We note that $p_{\rm f} (b) = \sum_a p(a,b)$.

\

\textit{Backward process.}
To formulate the backward process,  we need to introduce the time-reversal operator  $\Theta$ acting on $\bm H$, which is an anti-unitary (i.e., inner-product preserving and anti-linear) operator satisfying $\Theta^2 = \Theta$ and $\Theta^\dagger = \Theta$.  Here, an anti-linear operator satisfies that, for any $| \varphi_1 \rangle, | \varphi_2 \rangle \in \bm H$ and $\alpha_1, \alpha_2 \in \mathbb C$,
\begin{equation}
\Theta (\alpha_1 | \varphi_1 \rangle + \alpha_2 | \varphi_2 \rangle) = \alpha_1^\ast \Theta | \varphi_1 \rangle + \alpha_2^\ast  \Theta | \varphi_2 \rangle,
\end{equation}
where $\alpha_i^\ast$ means the complex conjugate of $\alpha_i$.  
Let $| \tilde \psi_a \rangle := \Theta | \psi_a \rangle$, $| \tilde \phi_b \rangle := \Theta | \phi_b \rangle$, and $\tilde \rho_0 := \Theta \rho_0 \Theta =  \sum_b p_0 (b) | \tilde \phi_b \rangle \langle \tilde \phi_b | $.  We note that $\{ | \tilde \psi_a \rangle \}$ and$\{ | \tilde \phi_b \rangle \}$ are orthonormal bases of $\bm H$.

The protocol for the backward process is as follows.
The initial state of the backward process is given by $\tilde \rho_0$.
 We first perform the projection measurement on $\tilde \rho_0$ with basis $\{ | \tilde \phi_b \rangle \}$, and obtain outcome $b$ with probability $p_0(b)$.  By this measurement, the ensemble average of the post-measurement states equals $\tilde \rho_0$.  
We introduce the time-reversal of the Hamiltonian as
\begin{equation}
\tilde H (t) := \Theta H (t) \Theta.
\label{Hamiltonian_symmetry}
\end{equation}
For example, if the Hamiltonian depends on magnetic field $B$ as $H(t; B)$,  its time-reversal is given by $\tilde H(t; B) =  H(t; -B)$.
We next perform a unitary operation from $t=0$ to $t= \tau$ with the time-reversed control protocol of the time-reversed Hamiltonian. 
The corresponding unitary operator $\tilde U$  is given by
\begin{equation} 
\tilde U := {\rm T}\exp \left( -{\rm i} \int_0^\tau \tilde H(\tau - t) dt \right).
\end{equation}
We next perform the projection measurement on $\tilde U \tilde \rho_0 \tilde U^\dagger$ with basis $\{ | \tilde \psi_a \rangle \}$, and obtain outcome $a$ with probability $\tilde p_{\rm f} (a) := \langle \tilde \psi_a | \tilde U \tilde \rho_0 \tilde U^\dagger | \tilde \psi_a \rangle$.  The joint probability of $(b,a)$ in the backward process, denoted by $\tilde p (b,a)$, is then given by
\begin{equation}
\tilde p (b,a) = \tilde p (a \leftarrow b) p_0 (b),
\end{equation}
where
\begin{equation}
\tilde p (a \leftarrow b) = | \langle \tilde \psi_a | \tilde U | \tilde \phi_b \rangle |^2
\end{equation}
is the backward transition probability.  We note that $\tilde p_{\rm f}(a) := \sum_b \tilde  p(b,a)$.

\

We define the following quantity:
\begin{equation}
\sigma (a,b) := \ln \frac{p(a,b)}{\tilde p(b,a)},
\label{fluctuation_theorem}
\end{equation}
which is referred to as the stochastic entropy production in the forward process.  The average of the entropy production is given by
\begin{equation}
\langle \sigma \rangle = \sum_{a,b} p(a,b) \ln \frac{p(a,b)}{\tilde p(b,a)},
\label{entropy_production}
\end{equation}
which is positive because of the positivity of the classical relative entropy:
\begin{equation}
\langle \sigma \rangle \geq 0.
\label{second_fluctuation1}
\end{equation}

We discuss the relationship between inequality~(\ref{second_fluctuation1}) and inequality~(\ref{second_positive}) in Theorem~5.2 in Sec.~5.2.  The following theorem plays a key role.

\begin{theorem}
The classical relative entropy~(\ref{entropy_production}) can be written as 
\begin{equation}
\langle \sigma \rangle = S(\rho_{\rm f} \| \rho_0),
\label{entropy_quantum_relative}
\end{equation}
where $S(\rho_{\rm f} \| \rho_0)$ is the quantum relative entropy. 
\end{theorem}

\begin{proof}
We first note that $\Theta \tilde U \Theta = U^\dagger$ holds, because $\Theta ({\rm i}\tilde H (t)) \Theta = - {\rm i} H(t)$ holds.  We then have the key observation that the unitary evolution has a time-reversal symmetry: 
\begin{equation}
\begin{split}
&\tilde p (a \leftarrow b) = | \langle \tilde \psi_a | \tilde U | \tilde \phi_b\rangle |^2 = | \langle \psi_a | \Theta \tilde U  \Theta | \phi_b \rangle |^2 \\
&=  | \langle \psi_a |  U^\dagger | \phi_b \rangle |^2 = | \langle \phi_b |  U | \psi_a \rangle |^2 = p(b \leftarrow a).
\end{split}
\end{equation}
Therefore, we obtain
\begin{equation}
\sigma (a,b) = \ln \frac{p_{\rm i} (a)}{p_0(b)}, 
\end{equation}
which leads to
\begin{equation}
\langle \sigma \rangle = \sum_{a,b} p(a,b) \ln \frac{p_{\rm i} (a)}{p_0(b)} = \sum_a p_{\rm i}(a) \ln p_{\rm i} (a) - \sum_b p_{\rm f}(b) \ln p_0(b).
\end{equation}
Obviously,
\begin{equation}
\sum_a p_{\rm i}(a) \ln p_{\rm i} (a) = -S(\rho_{\rm i}) = -S(\rho_{\rm f}).
\label{f1}
\end{equation}
We also obtain
\begin{equation}
\begin{split}
&\sum_b p_{\rm f}(b) \ln p_0(b) = \sum_b \langle \phi_b | \rho_{\rm f} | \phi_b \rangle \ln p_0(b) \\
&= \sum_b \langle \phi_b | \rho_{\rm f} \ln \rho_0 | \phi_b \rangle = {\rm tr} [\rho_{\rm f} \ln \rho_0]. 
\end{split}
\label{f2}
\end{equation}
By combining Eqs.~(\ref{f1}) and (\ref{f2}), we obtain Eq.~(\ref{entropy_quantum_relative}).
$\Box$\end{proof}

From Eqs.~(\ref{second_positive0}) and (\ref{fluctuation_theorem}), we obtain
\begin{equation} 
\langle \sigma \rangle = - {\rm tr} [\rho_{\rm f} \ln \rho_0] - S(\rho_{\rm i}). 
\label{sigma_entropy}
\end{equation}
Therefore,  inequality~(\ref{second_fluctuation1}) is equivalent to inequality~(\ref{second_positive}).  Since inequality~(\ref{second_positive}) leads to inequalities (\ref{second_positive3}) and (\ref{second_positive5})  as special cases, these inequalities can also be regarded as special cases of inequality~(\ref{second_fluctuation1}), which we will discuss in detail later.

Equality~(\ref{fluctuation_theorem})  is  regarded as a general expression of the quantum fluctuation theorem.    The reason why Eq.~(\ref{fluctuation_theorem}) can be called a ``theorem'' rather than just a definition lies in the fact that $\sigma$  equals to some important thermodynamic quantities for special cases, as we will show in Examples~5.3 and 5.4. 
Strictly speaking, Eq.~(\ref{fluctuation_theorem}) should be called a theorem only for such cases with the thermodynamic  expressions of $\sigma$. 
In fact, we have already shown Eq.~(\ref{sigma_entropy}), which implies that the average of $\sigma$ reduces to the lhs's of (\ref{second_positive2}) and (\ref{second_positive5}).

Before going to such special cases, we discuss the some properties of $\sigma$ on the basis of Eq.~(\ref{fluctuation_theorem}).
To do so, we introduce the entropy production in the backward process as
\begin{equation}
\tilde \sigma (b,a) := \ln \frac{\tilde p(b,a)}{p(a,b)}.
\label{entropy_production_backward}
\end{equation}
In the backward process, $\Theta \rho_0 \Theta$ and $\Theta \rho_{\rm i} \Theta$ respectively play the roles of $\rho_{\rm i}$ and $\rho_0$ in the forward process.  Therefore,  definition (\ref{entropy_production_backward}) in the backward process is consistent with the definition (\ref{entropy_production}) in the forward process.  We note that
\begin{equation}
\sigma (a,b) = - \tilde \sigma (b,a).
\end{equation}
We  introduce the probability distribution of $\sigma$ as
\begin{equation}
p(\sigma = \Sigma) := \sum_{a,b} p(a,b) \delta ( \Sigma, \sigma (a,b)),
\end{equation}
where $\delta (\cdot, \cdot)$ is the Kronecker delta, and that of $\tilde \sigma$ as
\begin{equation}
\tilde p (\tilde \sigma = \Sigma) := \sum_{b,a} \tilde p(b,a) \delta ( \Sigma, \tilde \sigma (b,a)).
\end{equation}
We can show that 
\begin{equation}
\frac{\tilde p(\tilde \sigma = -\Sigma)}{p(\sigma = \Sigma)} = e^{-\Sigma},
\label{fluctuation_theorem2}
\end{equation}
because
\begin{equation}
\begin{split}
\tilde p (\tilde \sigma = - \Sigma)& = \sum_{a,b} \tilde p(b,a) \delta ( - \Sigma, \tilde \sigma (b,a))\\
&=  \sum_{a,b} p(a,b) e^{\tilde \sigma (b,a)} \delta ( - \Sigma, \tilde \sigma (b,a)) \\
&= e^{-\Sigma} \sum_{a,b} p(a,b)  \delta ( - \Sigma, \tilde \sigma (b,a)) \\
&= e^{-\Sigma} \sum_{a,b} p(a,b)  \delta ( \Sigma, \sigma (a,b)) \\
&= e^{-\Sigma}p(\sigma = \Sigma).
\end{split}
\end{equation}
We  also refer to Eq.~(\ref{fluctuation_theorem2}) as the quantum fluctuation theorem.
We can show that
\begin{equation}
\langle e^{-\sigma} \rangle = 1,
\label{fluctuation_theorem3}
\end{equation}
because
\begin{equation}
\langle e^{-\sigma} \rangle := \sum_\Sigma p(\sigma = \Sigma) e^{-\Sigma} = \sum_\Sigma \tilde p (\tilde \sigma = - \Sigma) = 1.
\end{equation}
Equality~(\ref{fluctuation_theorem3}) is called the integral fluctuation theorem or the quantum Jarzynski equality.
By using the Jensen inequality for the exponential function (i.e., $e^{-\langle \sigma \rangle} \leq \langle e^{-\sigma} \rangle$), we reproduce inequality~(\ref{second_fluctuation1}) from Eq.~(\ref{fluctuation_theorem3}).
We note that the quantum fluctuation theorems~(\ref{fluctuation_theorem}), (\ref{fluctuation_theorem2}), and (\ref{fluctuation_theorem3}) were obtained by Kurchan~\cite{Kurchan} and Tasaki~\cite{Tasaki1}.

\begin{example}
We consider the case of Example~5.1 in which $\rho_{\rm i}$ and $\rho_0$ are given by Eqs.~(\ref{positive1}) and (\ref{positive2}), respectively.  
Let $H(0) = \sum_a E_a(0) | \psi_a \rangle \langle \psi_a |$ and $H(\tau) = \sum_b E_b(\tau) | \phi_b \rangle \langle \phi_b |$ be the spectrum decompositions of the initial and final Hamiltonians, which leads to $p_{\rm i}(a) = e^{\beta( F(0) - E_a(0))}$ and $p_0 (b) = e^{\beta ( F(\tau) - E_b(\tau) )}$.
The stochastic entropy production~(\ref{entropy_production}) is then given by
\begin{equation}
\sigma (a,b) = \ln \frac{p_{\rm i} (a)}{p_0(b)} = \beta (\Delta F - W(a,b)), 
\end{equation}
where $\Delta F := F(\tau) - F(0)$ and $W(a,b) := E_b (\tau ) - E_a (0)$.  We then obtain
\begin{equation}
\langle W \rangle := \sum_{a,b}p(a,b) W(a,b) = {\rm tr} [H(\tau)\rho_{\rm f}] - {\rm tr}[H(0) \rho_{\rm i}],
\end{equation}
which is consistent with Eq.~(\ref{work1}).
In this case, the integral fluctuation theorem (\ref{fluctuation_theorem3}) reduces to
\begin{equation} 
\langle e^{-\beta W} \rangle = e^{-\Delta F},
\end{equation}
which is called the quantum Jarzynski equality.
Inequality~(\ref{second_fluctuation1}) reduces to  (\ref{second_positive3})  in this situation.
\end{example}

\begin{example}
We next consider the case of Example~5.2 in which $\rho_{\rm i}$ and $\rho_0$ are given by Eqs.~(\ref{initial1}) and (\ref{ref1}), respectively. Let  $\rho_{\rm i}^S := \sum_{a'} p_{\rm i}^S (a') | \psi_{a'}^S \rangle \langle \psi_{a'}^S |$ and $\rho_{\rm f} := \sum_{b'} p_{\rm f}^S (b') | \phi_{b'}^S \rangle \langle \phi_{b'}^S |$ be the spectrum decompositions of the initial and final states of $S$.  The spectrum decompositions of the Hamiltonians of the heat baths are given by  $H^{B_k} = \sum_{a_k} E_{a_k}^{B_k} | \varphi_{a_k} \rangle \langle \varphi_{a_k} |$.  We then  have
\begin{equation}
\begin{split}
p_{\rm i}(a) &= p_{\rm i}^S (a') \prod_k \exp ( \beta_k (F_k - E_{a_k}^{B_k}) ), \\
p_{0}(b) &= p_{\rm f}^S (b') \prod_k \exp ( \beta_k (F_k - E_{b_k}^{B_k}) ),
\end{split}
\end{equation}
where $a = (a', \{ a_k \})$ and $b = (b', \{ b_k \})$.  We note that $| \psi_a \rangle = | \psi_{a'}^S \rangle \otimes_k  | \varphi_{a_k} \rangle$ and $| \phi_b \rangle = | \psi_{b'}^S \rangle \otimes_k  | \varphi_{b_k} \rangle$.
Therefore, the stochastic entropy production is given by 
\begin{equation}
\sigma (a,b) = s_{\rm f}(b') - s_{\rm i}(a') - \sum_k \beta_k Q_k (a_k,b_k),
\end{equation}
where
\begin{equation}
s_{\rm i} (a') := - \ln p_{\rm i}^S (a'), \ s_{\rm f}(b')  := -\ln p_{\rm f}^S (b')
\end{equation}
are called the stochastic entropies of $S$, and
\begin{equation}
Q_k (a_k, b_k) := E_{a_k}^{B_k} - E_{b_k}^{B_k}
\end{equation}
is the heat absorbed by $S$ from $B_k$. We note that
\begin{equation}
\begin{split}
\langle s_{\rm i} \rangle  &:= \sum_{a,b} p(a,b) s_{\rm i}(a')   = - \sum_{a'} p_{\rm i}^S (a') \ln p_{\rm i}^S (a') = S(\rho_{\rm i}^S), \\
\langle s_{\rm f} \rangle  &:=  \sum_{a,b} p(a,b) s_{\rm f}(b')  = - \sum_{b'} p_{\rm f}^S (b') \ln p_{\rm f}^S (b') = S(\rho_{\rm f}^S),
\end{split}
\end{equation}
and
\begin{equation}
\langle Q_k \rangle := \sum_{a,b} p(a,b) Q_k (a_k, b_k) =  {\rm tr}\left[ H^{B_k} \rho_{\rm i} \right] - {\rm tr} \left[ H^{B_k}\rho_{\rm f} \right],
\end{equation}
which is consistent with Eq.~(\ref{heat1}).  Therefore, we obtain
\begin{equation}
\langle \sigma \rangle = \Delta S - \sum_k \beta_k \langle Q_k \rangle,
\end{equation}
where $\Delta S :=  S(\rho_{\rm f}^S) - S(\rho_{\rm i}^S)$.
Inequality~(\ref{second_fluctuation1}) reduces to the Clausius inequality (\ref{second_positive5})  in this situation.
\end{example}

We note that the quantum fluctuation theorem has been generalized to nonunitary processes including quantum measurements~\cite{Hanggi2}.


\subsection{From the Monotonicity of the Quantum Relative Entropy}

We next apply the monotonicity of the quantum relative entropy to a derivation of the second law of thermodynamics~\cite{Spohn}.  While the obtained inequalities are mathematically not equivalent to the inequalities obtained in the previous two sections, their physical meanings are the same for special cases.
The inequalities in this subsection can also be applied to transitions between nonequilibrium steady states, which leads to a quantum version of the Hatano-Sasa inequality~\cite{Yukawa2}.

\subsubsection{Time-independent Control}

We first consider relaxation processes to steady states, in which the external parameters that we can control do not depend on time.  The following theorem plays a key role.

\begin{theorem}
Let $\mathcal E: L(\bm H) \to L(\bm H)$ be a CPTP map and $\rho_{\rm i} \in Q(\bm H)$ be an initial state.  We assume that $\mathcal E$ has a unique steady state $\rho_{\rm ss}$ satisfying $\mathcal E (\rho_{\rm ss}) = \rho_{\rm ss}$.  Then
\begin{equation}
\Delta S \geq - \sigma_{\rm ex}^B,
\label{Hatano_Sasa1}
\end{equation}
where $\Delta S := S(\mathcal E (\rho_{\rm i})) - S(\rho_{\rm i})$ and 
\begin{equation}
\sigma_{\rm ex}^B :=  {\rm tr} [\mathcal E (\rho_{\rm i})  \ln \rho_{\rm ss}] - {\rm tr} [\rho_{\rm i} \ln \rho_{\rm ss}].
\label{excess_entropy1}
\end{equation}
\end{theorem}

\begin{proof}
Inequality~(\ref{Hatano_Sasa1}) is obvious from the monotonicity of the relative entropy: $S(\rho_{\rm i} \| \rho_{\rm ss}) \geq S(\mathcal E (\rho_{\rm i}) \| \mathcal E (\rho_{\rm ss}))$ with $\mathcal E (\rho_{\rm ss}) = \rho_{\rm ss}$.
$\Box$\end{proof}

We apply inequality~(\ref{Hatano_Sasa1}) to relaxation processes to equilibrium states.

\begin{example}
Suppose that $\rho_{\rm ss} = I / d$ holds, where $I$ is the identity on $\bm H$ and $d$ is the dimension of $\bm H$.  In this case, $\sigma_{\rm ex}^B = 0$ holds, and therefore inequality~(\ref{Hatano_Sasa1}) reduces to
\begin{equation}
\Delta S \geq 0,
\label{second_monotone}
\end{equation}
which is equivalent to inequality~(\ref{entropy_increase}). 

Physically, the condition of $\mathcal E (I) = I$ implies that the steady state is the microcanonical distribution.  In fact, the microcanonical distribution is given by $\rho_{\rm microcan} := I /d$, where $\bm H$ is taken as the set of state vectors in a microcanonical energy shell.  Thus, inequality~(\ref{second_monotone})  is  regarded as the law of entropy increase for adiabatic processes, in which the system does not exchange the energy with the environment so that the steady state is expected to be the microcanonical distribution.
\end{example}

\begin{example}
Suppose that $\rho_{\rm ss} = e^{\beta (F- H)} $ holds, where $H$ is the Hamiltonian and $F :=  -\beta^{-1} \ln {\rm tr}[e^{-\beta H}]$ is the corresponding free energy.  In this case, $\sigma_{\rm ex}^B$ is given by
\begin{equation}
\sigma_{\rm ex}^B = - \beta (  {\rm tr} [H \mathcal E(\rho_{\rm i})] -  {\rm tr} [H\rho_{\rm i}]) = -\beta  \langle Q \rangle,
\end{equation}
where $Q$ is the heat absorbed by the system.  Therefore, inequality~(\ref{Hatano_Sasa1}) reduces to~\cite{Spohn}
\begin{equation}
\Delta S \geq \beta Q,
\label{second_monotone1}
\end{equation}
which is the Clausius inequality with a single heat bath.
\end{example}

We note that inequality~(\ref{Hatano_Sasa1}) can be applied to situations in which $\rho_{\rm ss}$ describes a nonequilibrium steady state (NESS).  However,  in such a case, the physical meaning of $\sigma_{\rm ex}^B$ is not so clear.  For the case of a classical overdamped Langevin system, Hatano and Sasa~\cite{Hatano} showed  that  $- \sigma_{\rm ex}^B$ can be regarded as an ``excess heat,'' which is obtained by subtracting a ``housekeeping heat'' from the total heat.  The housekeeping heat means the heat current  in a NESS, which vanishes for the case of an equilibrium steady state.


\subsubsection{Time-dependent Control}

We next consider situations in which  we  drive a system by changing  external parameters.  We assume that we change the values of the parameters $N-1$ times during the entire time evolution.   In such a case, the total time evolution $\mathcal E$ can be written as
\begin{equation}
\mathcal E = \mathcal E_N \circ \cdots \circ \mathcal E_2 \circ \mathcal E_1,
\end{equation}
where $\mathcal E_n  : L(\bm H) \to L(\bm H)$ ($n=1,2,\cdots,N$) describes the time evolution with the $n$th values of the external parameters.  We assume that each $\mathcal E_n$ has a unique steady state  $\rho_{{\rm ss}, n}$. 
We write $\rho_n := \mathcal E_n (\rho_{n-1})$ and $\rho_0 := \rho_{\rm i}$, where $\rho_{\rm i} \in Q(\bm H)$ is the initial state.  From inequality~(\ref{Hatano_Sasa1}), we have a set of inequalities 
\begin{equation}
S(\rho_{n+1}) - S(\rho_n) \geq - \left( {\rm tr} [\rho_{n}  \ln \rho_{{\rm ss}, n}] - {\rm tr} [\rho_{n-1} \ln \rho_{{\rm ss}, n}] \right) \ (1 \leq n \leq N).
\end{equation}
By summing up them, we have the following theorem.

\begin{theorem}[Quantum Hatano-Sasa inequality]
\begin{equation}
\Delta S \geq - \sigma_{\rm ex}^B,
\label{Hatano_Sasa2}
\end{equation}
where
\begin{equation}
\sigma_{\rm ex}^B := \sum_{n=1}^{N} \left( {\rm tr} [\rho_{n}  \ln \rho_{{\rm ss}, n}] - {\rm tr} [\rho_{n-1} \ln \rho_{{\rm ss}, n}] \right).
\label{excess_entropy2}
\end{equation}
\end{theorem}

If $\rho_{{\rm ss},n}$'s are out of equilibrium, inequality (\ref{Hatano_Sasa2}) is regarded as a quantum version of the Hatano-Sasa inequality, which was obtained by Yukawa~\cite{Yukawa2}.
We note that Eq.~(\ref{excess_entropy2}) can be rewritten as 
\begin{equation}
\sigma_{\rm ex}^B = {\rm tr} [\rho_N \ln \rho_{{\rm ss}, N}]  - {\rm tr} [\rho_0 \ln \rho_{{\rm ss}, 1}] - \sum_{n=1}^{N-1} \left( {\rm tr} [\rho_{n}  \ln \rho_{{\rm ss}, n+1}] - {\rm tr} [\rho_{n} \ln \rho_{{\rm ss}, n}] \right).
\end{equation}
If the initial and final states are the steady states such that $\rho_0 = \rho_{{\rm ss}, 1}$ and $\rho_N = \rho_{{\rm ss}, N}$, we have $\sigma_{\rm ex}^B = - \Delta S - \sum_{n=1}^{N-1} \left( {\rm tr} [\rho_{n}  \ln \rho_{{\rm ss}, n+1}] - {\rm tr} [\rho_{n} \ln \rho_{{\rm ss}, n}] \right)$.  In this case, inequality~(\ref{Hatano_Sasa1}) reduces to
\begin{equation}
 \sum_{n=1}^{N-1} \left( {\rm tr} [\rho_{n}  \ln \rho_{{\rm ss}, n+1}] - {\rm tr} [\rho_{n} \ln \rho_{{\rm ss}, n}] \right) \leq 0.
 \label{second_monotone2}
\end{equation}
If $\rho_{{\rm ss}, n}$'s are equilibrium states, we again have the second law of thermodynamics as follows.

\begin{example}
Suppose that $\rho_{{\rm ss}, n} = e^{\beta (F_n -  H_n)}$ holds, where  $H_n $ is the  Hamiltonian  during $\mathcal E_n$ and $F_n :=  - \beta^{-1} \ln {\rm tr}[e^{-\beta H_n}]$ is the corresponding free energy.  In this case, $\sigma_{\rm ex}^B$ is given by
\begin{equation}
\sigma_{\rm ex}^B = - \beta \sum_{n=1}^N (  {\rm tr} [H_n \rho_{n}] -  {\rm tr} [H_n \rho_{n-1}]) =: -\beta \langle Q \rangle,
\end{equation}
where $\langle Q \rangle$ is the heat absorbed by the system.  Therefore, inequality~(\ref{Hatano_Sasa1}) again reduces to 
\begin{equation}
\Delta S \geq \beta \langle Q \rangle.
\label{second_monotone3}
\end{equation}
We note that inequality~(\ref{second_monotone3}) is not equivalent to inequality~(\ref{second_positive5}) with a single heat bath, because their setting are mathematically different.  However, their physical meanings are physically the same; the both describe the Clausius inequality in the presence of a single heat bath, where entropy $S$ is identified to the von Neumann entropy.
We note that
\begin{equation}
\begin{split}
&\sum_{n=1}^{N-1} \left( {\rm tr} [\rho_{n}  \ln \rho_{{\rm ss}, n+1}] - {\rm tr} [\rho_{n} \ln \rho_{{\rm ss}, n}] \right) \\
&= \beta (\Delta F - \langle W \rangle),
\end{split}
\end{equation}
where $\Delta F := F_N - F_1$ is the free-energy difference, and 
\begin{equation}
\langle W \rangle := \sum_{n=1}^{N-1} \left( {\rm tr} [\rho_{n}  H_{n+1}] - {\rm tr} [\rho_{n} H_n] \right)
\end{equation}
is the work performed on the system.  If the initial and final states are the canonical distributions such that $\rho_0 = \rho_{{\rm ss}, 1}$ and $\rho_N = \rho_{{\rm ss}, N}$, inequality~(\ref{second_monotone2}) leads to
\begin{equation}
\langle W \rangle \geq \Delta F.
\label{HSwork}
\end{equation}
We note that inequality (\ref{HSwork}) still holds when the final state $\rho_N $ is out of equilibrium.
\end{example}


\section{Second Law with Quantum Feedback Control}

We next discuss a generalization of the second law of thermodynamics with quantum feedback control~\cite{Lloyd1,Lloyd2,Nielsen,Zurek1,Kieu,Allahverdyan,Quan1,Sagawa-Ueda1,Jacobs,SWKim,Dong,Morikuni,Abreu2,Lahiri,Lu,Funo0,Funo,Park},  on the basis of the positivity and the monotonicity in parallel to Sec.~5. 
  Here, quantum feedback control~\cite{Wiseman} means that the control protocol on a system depends on an outcome of a quantum measurement on the system. 
The obtained generalization includes the QC-mutual information.  We note that the quantum fluctuation theorem with feedback control has been discussed in Refs.~\cite{Morikuni,Lahiri,Funo}, while the classical one has been discussed in Refs.~\cite{Sagawa-Ueda3,Suzuki,Horowitz2,Toyabe3,Ito,Sagawa-Ueda4,Sagawa-Ueda6,Sagawa-Ueda7,Ito2}.

\subsection{From the Positivity of the Quantum Relative Entropy}

We first derive a generalization of Theorem~5.2 with quantum feedback control.  
Let $\rho_{\rm i}  \in Q(\bm H)$ be the initial state of the system.
We first perform a unitary operation $U$ so that the system evolves to $\rho' := U \rho_{\rm i} U^\dagger$.  We next perform a quantum measurement with the set of Kraus operators $\{ M_b \}_{b \in B} \subset L(\bm H)$, where $B$ is the finite set of outcomes. The corresponding POVM $\{ E_b \}_{b \in B}$ is assumed to be given by
\begin{equation} 
E_b = M_b^\dagger M_b.
\end{equation}
The probability of outcome $b$ is $p(b) = {\rm tr}[E_b \rho']$, and the post-measurement state with outcome $b$ is $\rho'(b) := M_b \rho' M_b^\dagger / p(b)$.  The QC-mutual information for this measurement is given by
\begin{equation}
I_{\rm QC} = S(\rho') - \sum_b p(b) S(\rho'(b)).
\end{equation}
We then perform a unitary operation $U_b$ that depends on outcome $b$, which is the feedback control.  The final state with outcome $b$ is  $\rho_{\rm f} (b) := U_b \rho'(b) U_b^\dagger$, whose ensemble average is  $\rho_{\rm f} := \sum_b U_b M_b U \rho_{\rm i} U^\dagger M_b^\dagger U_b^\dagger$.
In this setup, Theorem~5.2 is generalized to the following theorem.

\begin{theorem}
Let $\rho_0 (b)$ be a reference state corresponding to outcome $b$.  Then
\begin{equation}
-\sum_b p(b) {\rm tr} [ \rho_{\rm f} (b) \ln \rho_0 (b) ] - S(\rho_{\rm i}) = \sum_b p(b) S(\rho_{\rm f}(b) \| \rho_0 (b)) - I_{\rm QC}
\label{second_feedback0}
\end{equation}
holds, which leads to
\begin{equation}
-\sum_b p(b) {\rm tr} [ \rho_{\rm f} (b) \ln \rho_0 (b) ] - S(\rho_{\rm i}) \geq - I_{\rm QC}.
\label{second_feedback}
\end{equation}
\end{theorem}

\begin{proof}
We straightforwardly calculate that
\begin{equation}
\begin{split}
&\sum_b p(b) S(\rho_{\rm f}(b) \| \rho_0 (b)) \\
&= -\sum_b p(b) \left( S(\rho_{\rm f}(b)) + {\rm tr} [ \rho_{\rm f} (b) \ln \rho_0 (b) ] \right) \\
&=-\sum_b p(b){\rm tr} [ \rho_{\rm f} (b) \ln \rho_0 (b) ]   - S(\rho_{\rm i}) + S(\rho_{\rm i}) -  \sum_b p(b)  S(\rho_{\rm f}(b)) \\
&= -\sum_b p(b){\rm tr} [ \rho_{\rm f} (b) \ln \rho_0 (b) ]   - S(\rho_{\rm i})  + S(\rho') - \sum_b p(b)  S(\rho'(b)) \\
&= -\sum_b p(b){\rm tr} [ \rho_{\rm f} (b) \ln \rho_0 (b) ]   - S(\rho_{\rm i}) + I_{\rm QC},
\end{split}
\end{equation}
which implies Eq.~(\ref{second_feedback0}).
$\Box$\end{proof}

\begin{example}
We apply the setup of Example~5.1 to the above theorem.
Let $H_{\rm i}$ and $H_{\rm f}(b)$ be the initial and final Hamiltonians, where the final one can depend on outcome $b$.  We assume that the initial state is the canonical distribution $\rho_{\rm i} = e^{\beta (F_{\rm i} - H_{\rm i})}$ with $F_{\rm i} := -\beta^{-1} \ln {\rm tr} [e^{-\beta H_{\rm i}}]$, and that the reference state is also  the canonical distribution $\rho_{0}(b) = e^{\beta (F_{\rm f}(b) - H_{\rm f}(b))}$ with $F_{\rm f}(b) := -\beta^{-1} \ln {\rm tr} [e^{-\beta H_{\rm f}(b)}]$.  We then have
\begin{equation}
-\sum_b p(b) {\rm tr} [ \rho_{\rm f} (b) \ln \rho_0 (b) ] - S(\rho_{\rm i})  = \beta \langle W - \Delta F  \rangle,
\end{equation}
where
\begin{equation}
\langle W \rangle := \sum_b p(b) {\rm tr}[H_{\rm f}(b) \rho_{\rm f}(b)] - {\rm tr}[H_{\rm i} \rho_{\rm i}]
\end{equation}
is the average of the work performed on the system, and
\begin{equation}
\langle \Delta F \rangle := \sum_b p(b) F_{\rm f}(b) - F_{\rm i}
\end{equation}
is the average of the free-energy difference.
Here, we assumed that  the energy change of the system during the measurement is the work.  Physically, this assumption implies that the measurement process is  adiabatic.
Inequality~(\ref{second_feedback}) then reduces to
\begin{equation}
\beta \langle W - \Delta F  \rangle \geq -I_{\rm QC},
\label{second_QC}
\end{equation}
which is the generalization of inequality~(\ref{second_positive3}) to feedback-controlled processes.

Inequality~(\ref{second_QC}) identifies the upper bound of the sum of the extractable work  $- \langle W \rangle$ and the free-energy gain $\langle \Delta F \rangle$, which is proportional to the QC-mutual information as $\beta^{-1}I_{\rm QC}$.  In the classical limit, the upper bound is given by the classical mutual information.  The equality in~(\ref{second_QC}) is achieved by a classical model called the Szilard engine~\cite{Szilard}, where $I_{\rm QC} = \ln 2$, $\- \langle W \rangle = \beta^{-1} \ln 2$, and $\langle \Delta F \rangle = 0$.  Several models that achieves the equality in (\ref{second_QC}) have been discussed for both classical~\cite{Sagawa-Ueda4,Abreu,Horowitz3,Horowitz5} and quantum~\cite{Jacobs} regimes.  We note that inequality~(\ref{second_QC}) was first obtained by Sagawa and Ueda~\cite{Sagawa-Ueda1}.
\end{example}

We note that, in a similar manner to the above example, we can apply the setup of Example.~5.2 to Theorem~6.1, and obtain
\begin{equation}
\Delta S - \sum_k \beta_k \langle Q_k \rangle \geq - I_{\rm QC}.
\end{equation}


\subsection{From the Monotonicity of the Quantum Relative Entropy}

We next derive a generalization of the quantum Hatano-Sasa inequality~(\ref{Hatano_Sasa2})  (Theorem~5.5) with quantum feedback control, on the basis of the monotonicity of the quantum relative entropy.
Let $\rho_{\rm i} \in Q(\bm H)$ be the initial state.
We control the system through external parameters without feedback.
The system then evolves by a CPTP map $\mathcal E_m \circ \cdots \circ \mathcal E_2 \circ \mathcal E_1$, where each $\mathcal E_n$ ($n=1,2, \cdots, m$) has a unique steady state $\rho_{{\rm ss}, n}$.  We define
\begin{equation}
\sigma_{\rm ex}^B (1 \to m) := \sum_{n=1}^{m} \left( {\rm tr} [\rho_{n}  \ln \rho_{{\rm ss}, n}] - {\rm tr} [\rho_{n-1} \ln \rho_{{\rm ss}, n}] \right),
\label{excess_entropy3}
\end{equation}
where $\rho_n := \mathcal E_n (\rho_{n-1})$ and $\rho_0 := \rho_{\rm i}$. 

 We next perform  a measurement on $\rho_m$ with  instrument $\{ \mathcal E_{\rm meas}^{(b)} \}_{b \in B}$, where $B$ is the finite set of possible outcomes.  The corresponding Kraus representation is assumed to be given by
\begin{equation}
\mathcal E_{\rm meas}^{(b)} (\rho_m) :=  M_b \rho_m M_b^\dagger,
\end{equation}
where $M_b$'s are  Kraus operators.  
We note that $\sum_b \mathcal E_{\rm meas}^{(b)}$ is trace-preserving and that  $\sum_{b} M_b^\dagger M_b = I$ holds with the identity $I \in L(\bm H)$.  
The corresponding POVM $\{ E_b \}_{b \in B}$ satisfies  $E_b = M_b^\dagger M_b$ for any $b$.  The probability of obtaining outcome $b$ is  $p(b) = {\rm tr}[E_b \rho_m]$, and the post-measurement state with outcome $b$ is  $\rho_m(b) := M_b \rho_m M_b^\dagger / p(b)$.  The QC-mutual information for this measurement is  given by
\begin{equation}
I_{\rm QC} = S(\rho_m) - \sum_b p(b) S(\rho_m(b)).
\label{QC1}
\end{equation}
We define $\sigma_{\rm ex}^B = 0$ during the measurement.  Physically, this definition implies that the measurement process is assumed to be adiabatic, which is consistent with the assumption of the adiabaticity  in Example~6.1.

We next  perform feedback control with a CPTP map $\mathcal E_{N}^{(b)} \circ \cdots \circ \mathcal E_{m+2}^{(b)} \circ \mathcal E_{m+1}^{(b)}$ that depends on outcome $b$.   We assume that each CPTP map $\mathcal E_n^{(b)}$ ($n= m+1, \cdots, N$) has a unique steady state $\rho_{{\rm ss}, n}^{(b)}$.  We define
\begin{equation}
\sigma_{\rm ex}^B (m+1 \to N; b) := \sum_{n=m+1}^{N} \left( {\rm tr} [\rho_{n}^{(b)}  \ln \rho_{{\rm ss}, n}^{(b)}] - {\rm tr} [\rho_{n-1}^{(b)} \ln \rho_{{\rm ss}, n}^{(b)}] \right),
\label{excess_entropy3}
\end{equation}
where $\rho_n^{(b)} := \mathcal E_n^{(b)} (\rho_{n-1}^{(b)})$ ($n= m+1, \cdots, N$), and define
\begin{equation}
\sigma_{\rm ex}^B := \sigma_{\rm ex}^B (1 \to m) + \sum_b p(b)  \sigma_{\rm ex}^B (m+1 \to N; b). 
\end{equation}
We note that the ensemble average of the time evolution is given by
\begin{equation}
\mathcal E = \sum_b \mathcal E_{N}^{(b)} \circ \cdots \circ \mathcal E_{m+1}^{(b)} \circ \mathcal E_{\rm meas}^{(b)} \circ \mathcal E_m \circ \cdots \circ \mathcal E_2 \circ \mathcal E_1,
\end{equation}
which is CPTP.  We write $\rho_{\rm f} := \mathcal E  (\rho_{\rm i}) = \sum_b p(b) \rho_N^{(b)}$.  Then, the quantum Hatano-Sasa inequality~(\ref{Hatano_Sasa2}) is generalized to the following theorem.

\begin{theorem}
In the above setup,
\begin{equation}
\sum_b p(b) S(\rho_N^{(b)}) - S(\rho_{\rm i}) \geq - \sigma_{\rm ex}^B - I_{\rm QC}.
\label{second_feedback2}
\end{equation}
\end{theorem}

\begin{proof}
We have a set of inequalities from the monotonicity of the relative entropy:
\begin{equation}
S(\rho_{n+1}) - S(\rho_n) \geq - \left( {\rm tr} [\rho_{n}  \ln \rho_{{\rm ss}, n}] - {\rm tr} [\rho_{n-1} \ln \rho_{{\rm ss}, n}] \right) \ (1 \leq n \leq m),
\end{equation}
\begin{equation}
S(\rho_{n+1}^{(b)}) - S(\rho_n^{(b)}) \geq \left( {\rm tr} [\rho_{n}^{(b)}  \ln \rho_{{\rm ss}, n}^{(b)}] - {\rm tr} [\rho_{n-1}^{(b)} \ln \rho_{{\rm ss}, n}^{(b)}] \right)  \ (m+1 \leq n \leq N).
\end{equation}
By combining them to Eq.~(\ref{QC1}), we obtain inequality~(\ref{second_feedback2}).
$\Box$\end{proof}

Inequality (\ref{second_feedback2}) has been obtained in this article for the first time.  On the other hand, the classical Hatano-Sasa equality and inequality have been generalized to feedback-controlled classical systems~\cite{Abreu2,Lahiri}.


\begin{example}
We apply the setup in Example~5.7 to Theorem~6.2.
Let $H_n $ be the Hamiltonian corresponding to $\mathcal E_n$ and $F_n :=  - \beta^{-1} \ln {\rm tr}[e^{-\beta H_n}]$ be the free energy ($1 \leq n \leq N$).   We note that $H_n$ and $F_n$ can depend on outcome $b$ for $n \geq m+1$, which are denoted as $H_n(b)$ and $F_n (b)$.
We assume that the steady states are  the canonical distributions: $\rho_{{\rm ss}, n} = e^{\beta (F_n -  H_n)}$ ($n \leq m$) and $\rho_{{\rm ss}, n}^{(b)} = e^{\beta (F_n(b) -  H_n(b))}$ ($n \geq m+1$).   In this case, $\sigma_{\rm ex}^B$ is given by
\begin{equation}
\begin{split}
\sigma_{\rm ex}^B &= - \beta \sum_{n=1}^m \left(  {\rm tr} [H_n \rho_{n}] -  {\rm tr} [H_n \rho_{n-1}] \right) \\
&- \beta \sum_b p(b) \sum_{n=m+1}^N \left(  {\rm tr} [H_n(b) \rho_{n}(b)] -  {\rm tr} [H_n(b) \rho_{n-1}(b)] \right) \\
&=: -\beta \langle Q \rangle,
\end{split}
\end{equation}
where $\langle Q \rangle$ is the average of the heat absorbed by the system. Here, we assumed that the measurement process is adiabatic such that $\langle Q \rangle =0$ holds during the measurement, which is consistent with the definition $\sigma_{\rm ex}^B = 0$ during $\mathcal E_{\rm meas}^{(b)}$.   Inequality~(\ref{second_feedback2}) then reduces to
\begin{equation}
\langle \Delta S \rangle \geq \beta \langle Q \rangle - I_{\rm QC},
\label{second_feedback3}
\end{equation}
where $\langle \Delta S \rangle := \sum_b p(b) S(\rho_N^{(b)}) - S(\rho_{\rm i})$.

We further assume that the initial and final states are the canonical distributions such that $\rho_{\rm i} = \rho_{{\rm ss}, 1}$ and $\rho_N^{(b)} = \rho_{{\rm ss}, N}^{(b)}$. We then have
\begin{equation}
\langle \Delta S \rangle = \beta \langle \Delta E - \Delta F \rangle,
\end{equation}
where 
\begin{equation}
\langle \Delta E \rangle := \sum_b p(b) {\rm tr} [H_N^{(b)}\rho_N^{(b)}] - {\rm tr}[H_1 \rho_{\rm i}]
\end{equation}
is the energy difference of the system, and $\langle \Delta F \rangle := \sum_b p(b) F_N(b) - F_1$. Therefore, we obtain from inequality (\ref{second_feedback3}) that
\begin{equation}
\beta \langle W - \Delta F \rangle \geq - I_{\rm QC},
\label{second_QC2}
\end{equation}
where we defined $\langle W \rangle := \langle \Delta E \rangle - \langle Q \rangle$ due to  the first law of thermodynamics.  
While inequality~(\ref{second_QC2}) is mathematically not equivalent to inequality~(\ref{second_QC}), their physical meanings are the same.
We note that inequality (\ref{second_QC2}) still holds when $\rho_N^{(b)}$'s are out of equilibrium.
\end{example}


\section{Concluding Remarks}

We have discussed several second law-like inequalities on the basis of the positivity and the monotonicity of the quantum relative entropy.  In Sec.~2, we discussed the basic concepts for quantum states and dynamics.  In particular, we discussed the complete positivity of time evolutions, and derived  the Kraus representation~(\ref{Kraus}).  In Sec.~3, we introduced the von Neumann entropy and the quantum relative entropy, and discussed their basic properties.  In particular, we proved the positivity of the quantum relative entropy.  We also discussed the monotonicity, which is proved in Appendix~B.  In Sec.~4, we have discussed the quantum mutual information and two related quantities.  One is the Holevo's $\chi$-quantity and the other is the QC-mutual information.  The two quantities characterize the accessible classical information encoded in quantum states, through the Holevo bound (\ref{Holevo})  and its dual inequality (\ref{Holevo_dual}). We also showed the data processing inequalities (\ref{inequality1}) and (\ref{inequality3}), which are direct consequences of the monotonicity of the quantum relative entropy.

Sections~5 is devoted to the derivations of the second law of thermodynamics and its generalization.  
In Sec.~5.1, we  discussed that  the von Neumann entropy and the thermodynamic entropy can be identified in canonical distributions.
  In Sec.~5.2, we derived a general inequality~(\ref{second_positive}) on the basis of the positivity.  Inequality~(\ref{second_positive}) leads to a well-known expression of the second law of thermodynamics~(\ref{second_positive5}) (i.e., the Clausius inequality) as a special case, by treating the total system including the baths as a unitary system.  
In Sec.~5.3, we discussed the quantum fluctuation theorem, which leads to the second law of thermodynamics~(\ref{second_fluctuation1}) that is equivalent to inequality~(\ref{second_positive}) in Sec.~5.2.
In Sec.~5.4, we derived a general inequality~(\ref{Hatano_Sasa2})  on the basis of the monotonicity, which is regarded as a quantum version of the Hatano-Sasa inequality.  This inequality leads to the second law of thermodynamics~(\ref{second_monotone3}) for a special case.
Our derivations are independent of the size of the system, and therefore the obtained inequalities can be applied to small thermodynamic systems such as quantum dots.  

\

We now discuss the physical meanings of the results in more detail.
We first compare the two derivations of the second law, which are based on the positivity and the monotonicity.
The derivation of inequality (\ref{second_positive5}) based on the positivity is quite universal, because we have essentially made only three assumptions: (i) the total system including the baths obeys a unitary evolution, (ii)  the baths are initially in the canonical distributions, and (iii) the system is initially not correlated to the baths.  
We stress that we did not assume anything about the intermediate and the final states of the total system.
The assumptions (i) and (ii) would be physically reasonable.  The assumption (iii) can be justified if the system is initially separated from the baths.

On the other hand, in the derivation of inequality~(\ref{second_monotone1}) based on the monotonicity, any bath is not explicitly included in our formulation, and therefore time evolutions are assumed to be nonunitary.  In this case,  the derivation requires a strong assumption that (iv) the steady state of each CPTP map $\mathcal E_n$ is a canonical distribution.  We stress that  the assumption (iv) is not satisfied in general  if we trace out the bath from the total system satisfying (i)-(iii).  Moreover, the derivation based on the monotonicity implicitly requires that (v) the outside bath needs to be refreshed (i.e., needs to be replaced by a new bath that is not correlated with the system) when the external parameters are changed, because each $\mathcal E_n$ is assumed to be CPTP.
On the other hand,  the monotonicity can be applied to the cases in which the steady states are not in thermal equilibrium, which is an advantage of the method of the monotonicity.  We note that the derivation based on the monotonicity can be generalized to continuous-time systems described by quantum master equations~\cite{Lindblad1,Breuer,Wiseman,Gardiner}.

Going back to the assumption (iii) for the derivation based on the positivity,  it has been pointed out that if the system is initially correlated with the baths,  inequality~(\ref{second_positive5})  needs modifications~\cite{Allahverdyan2,Horhammer,Jennings2}. For example, in the special case of inequality (\ref{second_positive7}), the free energy of the system needs to be renormalized~\cite{Hanggi1}.
On the other hand, the role of the initial correlation has been discussed in terms of the origin of the arrow of the time~\cite{Maccone,Jennings1}.

The assumption (ii) can be criticized in terms of the foundation of statistical mechanics.
In macroscopic systems, a thermal equilibrium state does not necessarily correspond to any canonical distribution from the microscopic point of view.  
In fact, it has been shown that even a pure state can behave as a thermal equilibrium state~\cite{Neumann2,Deutsch,Sredkicki,Tasaki2,Sugita,Lebowitz,Popescue,Reimann,Ikeda,Sugiura,Sugiura2}. 
Therefore, the assumption (ii) is  difficult to be rigorously justified for macroscopic baths.
The identification of the von Neumann entropy to the thermodynamic entropy fails in general, except for the cases in which a thermal equilibrium state corresponds to  a canonical distribution.
Experimentally, trapped ultracold atoms can relax to a thermal equilibrium state even if they are well separated from the environment.
We note that the detailed investigations of the relaxation processes of isolated quantum systems are now experimentally accessible~\cite{Kinoshita}. 

We next briefly mention a theoretical approach to derive the second law without assuming the initial canonical distribution~\cite{Lenard,Pusz,Tasaki3,Campisi2}.
Let $\rho \in Q(\bm H)$ be an initial state.  We consider a unitary evolution $U$  from $t=0$ to $\tau$ with a time-dependent Hamiltonian $H(t)$.  
We define the work performed on the system as $\langle W \rangle := {\rm tr}[H(\tau) U \rho U^\dagger] - {\rm tr}[H(0) \rho]$.  We also define the quasi-static work $\langle W \rangle_{\rm slow} := \lim_{\tau \to \infty}\langle W \rangle$, where $H(0)$ and $H(\tau)$ are fixed.
Let $H(0) = \sum_k E_k | \psi_k \rangle \langle \psi_k |$ be the spectrum decomposition of the initial Hamiltonian.
The crucial assumption on the initial state is that $\rho$ is  written as $\rho = \sum_k p(k) | \psi_k \rangle \langle \psi_k |$ with $p(k)$'s satisfying
\begin{equation}
p(k) \geq p(k') \ \ \ \ \ {\rm if} \ \ E_k \leq E_{k'}.
\label{assumption_Lenard}
\end{equation}
Then, Lenard~\cite{Lenard} proved that  
\begin{equation}
\langle W \rangle \geq \langle W \rangle_{\rm slow},
\label{Lenard}
\end{equation}
which is regarded as an expression of the second law of thermodynamics (see also Ref.~\cite{Tasaki3}). We note that the quasi-static work can be identified to the free-energy difference as $\langle W \rangle_{\rm slow} = \Delta F$.

Without assuming~(\ref{assumption_Lenard}), inequality~(\ref{Lenard}) is not satisfied in general.
For example, the microcanonical distribution does not satisfy assumption~(\ref{assumption_Lenard}), if  the entire Hilbert space is spanned by all $|\psi_k \rangle$'s and is not restricted to the microcanonical energy shell.
In fact, several counterexamples against inequality (\ref{Lenard}) have been discussed in classical systems with the microcanonical initial distribution~\cite{Sato,Marathe,Jarzynski5}.
We note that a generalized quantum fluctuation theorem for microcanonical initial distributions has been discussed~\cite{Talkner3}.  

The role of a quantum coherence in the initial state has been studied in terms of quantum heat engines~\cite{Scovil,Geusic,Bender,Scully1,Scully2,Quan2,Liberato}.  For example, if the initial state differs from the canonical distribution and involves an additional quantum coherence, the second law of thermodynamics like~(\ref{Lenard}) is not necessarily satisfied.

We now reach an observation that the choice of the initial state is crucial to derive the second law.  In fact, the canonical distribution is a special state that maximizes its von Neumann entropy among the states that have the same amount of the energy.  This special initial condition is a crucial reason why we can derive the second law even if the dynamics of the total system is reversible.
It would be difficult to rigorously justify the physical validity of the initial canonical distribution.
Therefore, the derivations of the second law in Sec.~5  are not  so satisfactory to understand the fundamental reason why the macroscopic world is irreversible in spite of the reversibility of the microscopic dynamics.  To understand the origin of the arrow of the time would be an interesting but quite difficult future challenge.  
We note that there has been an interesting  approach to this problem  in terms of the ``waiting-time typicality''~\cite{Tasaki4,Ikeda2,Tasaki5}, which is based on unitary dynamics with an initial pure state.

Thermodynamics of information processing is also an interesting topic.  While it would be easier than the aforementioned fundamental problem, it is also closely related to the foundation of thermodynamics and statistical mechanics.   This topic has been discussed by numerous researchers in terms of  the paradox of ``Maxwell's demon''~\cite{Demon,Maruyama,Maroney1,Sagawa-Ueda5}.  The demon can be formulated as an information processing device, and has been shown to be  consistent with the second law of thermodynamics~\cite{Maroney1,Sagawa-Ueda5,Maroney2,Sagawa-Ueda2}.  
 Recently, modern nonequilibrium statistical mechanics and quantum information theory have shed new light on the theory of thermodynamics including  the demon~\cite{Lloyd1,Lloyd2,Nielsen,Zurek1,Kieu,Allahverdyan,Quan1,Sagawa-Ueda1,Jacobs,SWKim,Dong,Morikuni,Abreu2,Lahiri,Lu,Funo0,Funo,Park,Demon,Maruyama,Maroney1,Sagawa-Ueda5,Maroney2,Sagawa-Ueda2,Sagawa-Ueda3,Suzuki,Horowitz2,Toyabe3,Ito,Sagawa-Ueda4,Sagawa-Ueda6,Sagawa-Ueda7,Ito2,Abreu,Horowitz3,Horowitz5}.  
 
In particular, we  discussed thermodynamics of quantum feedback control in Sec.~6.  We obtained inequalities~(\ref{second_feedback}) and (\ref{second_feedback2}) as generalizations of (\ref{second_positive}) and (\ref{Hatano_Sasa2}), respectively.  The obtained inequalities lead to generalizations of the second law (\ref{second_QC}) and (\ref{second_feedback3}).   We note that the quantum fluctuation theorem with quantum information processing has been discussed in Refs.~\cite{Morikuni,Lahiri,Funo}.


\appendix

\section{Short Summary of Linear Algebra}

In this appendix, we briefly summarize the concepts and notations of the linear algebra for finite-dimensional Hilbert spaces. 
 Let $\bm H$ be a $N$-dimensional Hilbert space with $N < \infty$.  The inner product of $| \varphi \rangle, | \phi \rangle  \in \bm H$ is written as $\langle \varphi | \phi \rangle$ or $(\varphi, \phi)$. 
The norm of the Hilbert space is given by $\| \varphi \|^2 = \langle \varphi | \varphi \rangle$.
A linear basis $\{ | \varphi_k \rangle \}_{k=1}^N \subset \bm H$ is called an orthonormal basis if it satisfies $\langle \varphi_k | \varphi_l \rangle = \delta_{kl}$ with $\delta_{kl}$ the Kronecker delta.

  Let $L (\bm H, \bm H')$ be the set of linear operators from $\bm H$ to another Hilbert space $\bm H'$.  In particular, we write $L(\bm H) := L(\bm H, \bm H)$.  We write $\langle \varphi | X | \phi \rangle := (\varphi, X\phi)$ for $X \in L(\bm H)$.
For any  $X \in L(\bm H, \bm H')$, there exists a unique operator $X^\dagger \in L(\bm H', \bm H)$ that satisfies 
\begin{equation}
(\phi, X \psi ) = (X^\dagger \phi, \psi) 
\end{equation}
for any $| \psi \rangle \in \bm H$ and $| \phi \rangle \in \bm H'$.  $X^\dagger$ is called the adjoint or the Hermitian conjugate operator of $X$.  If $X^\dagger = X$ holds for $X \in L(\bm H)$, $X$ is called self-adjoint or Hermitian.  We note that any operator $X \in L(\bm H)$ can be written as $X = X_1 + {\rm i} X_2$, where $X_1:= (X + X^\dagger) / 2$ and $X_2 := (X - X^\dagger)/(2{\rm i})$ are Hermitian.

Let $\bm H'$ be a subspace of $\bm H$ with dimension $N'$ ($\leq N$).  Let $\{ | \varphi_k \rangle \}_{k=1}^{N'}$ be an orthonormal basis of $\bm H'$.  The projection operator onto $\bm H'$ is given by
\begin{equation}
P_{\bm H'} = \sum_{k=1}^{N'} | \varphi_k \rangle \langle \varphi_k | \in L(\bm H).
\end{equation}
In particular, any orthonormal basis satisfies
\begin{equation}
\sum_{k=1}^{N} | \varphi_k \rangle \langle \varphi_k |  = I,
\end{equation}
where $I$ is the identity operator on $\bm H$.
On the other hand, $| \varphi \rangle \langle \varphi |$  with $\| \varphi \| =1$ is the projection operator onto the $1$-dimensional space that is spanned by  $| \varphi \rangle$.  
Any Hermitian operator $X$ has the spectrum decomposition 
\begin{equation}
X = \sum_k x_k | \varphi_k \rangle \langle \varphi_k |,
\label{spectrum}
\end{equation}
where $x_k \in \mathbb R$ is an eigenvalue and $| \varphi_k \rangle$ is an eigenvector that constitutes an orthonormal basis $\{ | \varphi_k \rangle \} \subset \bm H$.  The support of $X$ is the subspace of $\bm H$ that is spanned by $| \varphi_k \rangle$'s with nonzero eigenvalues.

If $\langle \varphi | X | \varphi \rangle \geq 0$ holds for any $| \varphi \rangle \in \bm H$, $X \in L(\bm H)$ is called positive.  Any positive operator is Hermitian, because $\langle \varphi | X_2 | \varphi \rangle$ should be zero  for any $| \varphi \rangle$ with  $X_2$ defined above. Any eigenvalue of a positive operator is non-negative. We write $X \geq Y$ if $X - Y$ is positive.  In particular, $X \geq 0$ if  $X$ is positive.  If $\langle \varphi | X | \varphi \rangle > 0$ holds for any nonzero $| \varphi \rangle \in \bm H$, $X \in L(\bm H)$ is called positive definite and is written as $X > 0$.  Any positive-definite operator is Hermitian with positive eigenvalues.  We note that the support of a positive-definite operator equals to $\bm H$.  
We also note that $V \in L(\bm H, \bm H')$ is called a contraction, if $\langle \varphi | V^\dagger V| \varphi \rangle  \leq \langle \varphi | \varphi \rangle$ holds for any $| \varphi \rangle \in \bm H$.

The trace of $X \in L(\bm H)$ is defined as
\begin{equation}
{\rm tr}[X] := \sum_k \langle \varphi_k | X | \varphi_k \rangle, 
\label{trace}
\end{equation}
where $\{ | \varphi_k \rangle \} \subset \bm H$ is an orthonormal basis.   We note that ${\rm tr}[X]$ is independent of the choice of the orthonormal basis.  If the spectrum decomposition of $X$ is given by Eq.~(\ref{spectrum}), its trace is given by ${\rm tr}[X]  = \sum_k x_k$.  
Let $X \in L(\bm H, \bm H')$ and $Y \in L(\bm H', \bm H)$.  We then have
\begin{equation}
{\rm tr} [YX] = {\rm tr} [XY],
\label{trace2}
\end{equation}
where the left  and right traces are on $\bm H$ and $\bm H'$, respectively.  In fact,
\begin{equation}
{\rm tr} [YX] = \sum_{kl} \langle \phi_k | Y| \psi_l \rangle \langle \psi_l | X | \phi_k \rangle =   \sum_{kl} \langle \psi_l | X | \phi_k \rangle  \langle \phi_k | Y | \psi_l \rangle = {\rm tr} [XY],
\end{equation}
where $\{ | \phi_k \rangle \}$ and $\{ | \psi_k \rangle \}$ are respectively orthonormal bases of $\bm H$ and $\bm H'$.

Let $\bm H$ and $\bm H'$ be Hilbert spaces with dimensions $N$ and $N'$, respectively.  
For any $X \in L(\bm H, \bm H')$, there exist orthonormal bases $\{ | \varphi_k \rangle \}_{k=1}^N \subset \bm H$ and $\{ | \psi_k \rangle \}_{k=1}^{N'} \subset \bm H'$ such that 
\begin{equation}
X = \sum_{k=1}^{N''} \lambda_k | \psi_k \rangle \langle \varphi_k |,
\label{singular_value}
\end{equation}
where $\lambda_k \geq 0$ ($1 \leq k \leq N''$) and $N'' := \min \{ N,N' \}$.  Equality~(\ref{singular_value}) is called the singular value decomposition, and $\lambda_k$ is called a singular value of $X$.

We next consider the tensor product of two Hilbert spaces.
Let $\bm H_A$ and $\bm H_B$ be two Hilbert spaces, and $\bm H_A \otimes \bm H_B$ be their tensor product.
For simplicity, we write $| \varphi_A \rangle \otimes | \varphi_B \rangle \in \bm H_A \otimes \bm H_B$ as $| \varphi_A \rangle  | \varphi_B \rangle$.  Let $\{ | \varphi_k \rangle \}$ and $\{ | \psi_l \rangle \}$ be orthonormal bases of $\bm H_A$ and $\bm H_B$, respectively.   Any vector $| \Psi \rangle \in \bm H_A \otimes \bm H_B$ can be written as
\begin{equation}
| \Psi \rangle = \sum_{kl} \alpha_{kl} | \varphi_k \rangle | \psi_l \rangle,
\end{equation}
where $\alpha_{kl} \in \mathbb C$.
By applying the singular-value decomposition to matrix $(\alpha_{kl})$, we find that there are orthonormal bases  $\{ | \varphi'_k \rangle \} \subset \bm H_A$ and $\{ | \psi'_l \rangle \} \subset \bm H_B$ such that
\begin{equation}
| \Psi \rangle = \sum_k \lambda_k  | \varphi'_k \rangle | \psi'_k \rangle,
\label{Schmidt}
\end{equation}
where $\lambda_k \geq 0$ is a singular value of matrix $(\alpha_{kl})$.
Equality~(\ref{Schmidt}) is called the Schmidt decomposition of $| \Psi \rangle$.

\section{Proof of the Monotonicity of the Quantum Relative Entropy}

In this appendix, we prove the monotonicity~(\ref{monotonicity}) of the quantum relative entropy  (Theorem~3.5) in line with Petz's proof~\cite{Petz3}.   We first prove some lemmas, in which the key concepts are the operator monotonicity and the operator convexity~\cite{Bhatia}.  

Let $\bm H$ be a finite-dimensional Hilbert space and  $X \in L(\bm H)$ be a positive-definite operator with spectrum decomposition $X = \sum_k x_k | \varphi_k \rangle \langle \varphi_k |$.   We can substitute $X$  to a function $f : (0, \infty) \to \mathbb R$  as
\begin{equation}
f(X) := \sum_k f(x_k) | \varphi_k \rangle \langle \varphi_k |.
\end{equation}
We then introduce the following concepts.

\begin{definition}
$f$ is called decreasing-operator monotone, if $f(X) \geq f(Y)$ holds for any positive-definite operators $X,Y \in L(\bm H)$ satisfying $X \leq Y$.
\end{definition}

\begin{definition}
$f$ is called operator convex, if $f(pX+(1-p)Y) \leq pf(X) + (1-p)f(Y)$ holds for any $0 \leq p \leq 1$ and any positive-definite operators $X,Y \in L(\bm H)$.
\end{definition}

In general,  it is difficult to judge whether a function is operator monotone and is operator convex.  However,  it is not so difficult to show that $f(x) = (x+t)^{-1}$ ($t \geq 0$) and $f(x) = -\ln x$ are both decreasing-operator monotone and operator convex.

\begin{lemma}
$f(x) = (x+t)^{-1}$ $(t \geq 0)$ is decreasing-operator monotone.
\end{lemma}

\begin{proof}
It is obvious that $X \leq I \Rightarrow X^{-1} \geq I$ with $I \in \bm H$ the identity.
We then have $Y^{-1/2}XY^{-1/2}  \leq I \Rightarrow Y^{1/2}X^{-1}Y^{1/2} \geq I$, and therefore $X  \leq Y \Rightarrow X^{-1} \geq Y^{-1}$.  Since $X \leq Y \Rightarrow X+tI \leq Y+tI$, we obtain $X \leq Y \Rightarrow (X+tI)^{-1} \geq (Y+tI)^{-1}$.
$\Box$\end{proof}

\begin{lemma}
$f(x) = (x+t)^{-1}$  $(t \geq 0)$  is operator convex.
\end{lemma}

\begin{proof}
Since $x^{-1}$ is convex, $(pX + (1-p) I)^{-1} \leq pX^{-1} + (1-p)I$ holds.   We then have $(pY^{-1/2}XY^{-1/2} + (1-p) I)^{-1} \leq pY^{1/2}X^{-1}Y^{1/2} + (1-p)I$, and therefore $(pX + (1-p) Y)^{-1} \leq pX^{-1}+ (1-p)Y^{-1}$.  By noting that $(pX + (1-p) Y + tI) = p(X + tI) + (1-p)(Y+tI)$, we obtain $(pX + (1-p) Y + tI)^{-1} \leq p(X + tI)^{-1}+ (1-p)(Y+tI)^{-1}$.
$\Box$\end{proof}

\begin{lemma}
$f(x) := - \ln x$ is decreasing-operator monotone and operator convex.
\end{lemma}

\begin{proof}
It follows from that $-\ln x = \int_0^\infty \left( (x+t)^{-1} - (1+t)^{-1} \right) dt$. $\Box$
\end{proof}

We next show an important property of operator convex functions.

\begin{lemma}[Jensen's operator inequality~\cite{Davis,Hansen1,Hansen2}]
If $f$ is operator convex and $\lim_{x \to +0} f(x) \leq 0$,  
\begin{equation}
f(V^\dagger XV) \leq V^\dagger f(X) V
\label{operator_convex2}
\end{equation}
holds for any  contraction $V \in L(\bm H, \bm H')$ and any positive-definite operator $X \in L(\bm H')$ such that $V^\dagger X V$ is also positive definite.\footnote{The definition of a contraction is given in Appendix~A.}
\end{lemma}

\begin{proof}
 We define
\begin{equation}
\begin{split}
X' &:= \left[
\begin{array}{cc}
X & 0 \\
0 & 0
\end{array}
\right]
\in L(\bm H' \oplus \bm H, \bm H' \oplus \bm H), \\
V_1 &:= \left[
\begin{array}{cc}
V & (I-VV^\dagger)^{1/2} \\
(I-V^\dagger V)^{1/2} & -V^\dagger
\end{array}
\right]
\in L(\bm H \oplus \bm H', \bm H' \oplus \bm H), \\
V_2 &:= \left[
\begin{array}{cc}
V & -(I-VV^\dagger)^{1/2} \\
(I-V^\dagger V)^{1/2} & V^\dagger
\end{array}
\right]
\in L(\bm H \oplus \bm H', \bm H' \oplus \bm H),
\end{split}
\end{equation}
where we used the assumption that $V$ is a contraction to define $(I-V^\dagger V)^{1/2}$ and $(I-VV^\dagger)^{1/2}$.
By noting that $\bm H \oplus \bm H' \simeq \bm H' \oplus \bm H$ and by using the singular-value decomposition of $V$, it is easy to check that $V_1$ and $V_2$ are unitary.  We have
\begin{equation}
V_1^\dagger X' V_1 = \left[
\begin{array}{cc}
V^\dagger XV & V^\dagger XU \\
UXV & UXU
\end{array}
\right], \ 
V_2^\dagger X' V_2 = \left[
\begin{array}{cc}
V^\dagger XV & - V^\dagger XU \\
-UXV & UXU
\end{array}
\right], 
\end{equation}
where  $U := (I - VV^\dagger)^{1/2}$, and  
\begin{equation}
\frac{V_1^\dagger X' V_1 + V_2^\dagger X' V_2}{2} = \left[
\begin{array}{cc}
V^\dagger XV &0 \\
0 & UXU
\end{array}
\right].
\end{equation}
By using the operator convexity of $f(x)$, we obtain
\begin{equation}
\begin{split}
&\left[
\begin{array}{cc}
f(V^\dagger XV) &0 \\
0 & f( UXU) 
\end{array}
\right] = f \left( \frac{V_1^\dagger X' V_1 + V_2^\dagger X' V_2}{2} \right) \leq  \frac{f(V_1^\dagger X' V_1) + f(V_2^\dagger X' V_2)}{2} \\
&= \frac{V_1^\dagger f( X')  V_1+ V_2^\dagger f( X') V_2}{2} = \left[
\begin{array}{cc}
V^\dagger f(X)V + U' f(0) U' &0 \\
0 &  Uf(X) U + V f(0) V^\dagger
\end{array}
\right],
 \end{split}
\end{equation}
where  $U' := (I - V^\dagger V)^{1/2}$ and $f(0) := \lim_{x \to +0} f(x) I$. From $f(0) \leq 0$, we finally obtain
\begin{equation}
\left[
\begin{array}{cc}
f(V^\dagger XV) &0 \\
0 & f( UXU) 
\end{array}
\right] \leq  \left[
\begin{array}{cc}
V^\dagger f(X)V &0 \\
0 &  Uf(X) U
\end{array}
\right], 
\end{equation}
which implies inequality (\ref{operator_convex2}).
$\Box$\end{proof}

Conversely, it is known that~\cite{Hansen1} if inequality (\ref{operator_convex2}) is satisfied for any contraction $V \in L(\bm H, \bm H')$ and any positive-definite operator $X \in L(\bm H')$, then $f$ is operator convex and $\lim_{x \to +0} f(x) \leq 0$.

While $-\ln x$ is operator convex from Lemma B.3, inequality~(\ref{operator_convex2}) does not hold for $f(x) = - \ln x$, because $-\ln x$ does not satisfy the assumption of $\lim_{x \to +0} f(x) \leq 0$ in Lemma B.4.

\

We next consider a generalization of the  Schwarz inequality.

\begin{lemma}[Kadison inequality~\cite{Kadison}]
Let  $X \in \bm H$ be a Hermitian operator and $\mathcal E: L(\bm H) \to L(\bm H')$ be a unital positive map.\footnote{The definition of a unital map is given in Sec.~2.2.2.}  Then
\begin{equation}
\mathcal E (X^2) \geq \mathcal E(X)^2.
\label{Kadison}
\end{equation}
\end{lemma}

\begin{proof}
Let $X := \sum_{k=1}^N x_k | \varphi_k \rangle \langle \varphi_k |$ be the spectrum decomposition of $X$, where $N$ is the dimension of $\bm H$.  We define  $X_k := \mathcal E( | \varphi_k \rangle \langle \varphi_k |)$, which satisfies $X_k \geq 0$  due to the positivity of $\mathcal E$ and  $\sum_k X_k = I$  due to assumption $\mathcal E ( I ) = I$.  Inequality~(\ref{Kadison}) can then be written as
\begin{equation}
\sum_k x_k^2 X_k \geq \left( \sum_k x_k X_k \right)^2,
\end{equation}
or equivalently, for any $| \psi \rangle \in \bm H'$,
\begin{equation}
\sum_k \langle \psi |  x_k^2 X_k | \psi \rangle \geq  \langle  z | z \rangle,
\label{Kadison1}
\end{equation}
where $| z \rangle := \sum_k x_k X_k | \psi \rangle$.   
We introduce  auxiliary system $\mathbb C^N$ that has an orthonormal basis $\{ | e_k \rangle \}_{k=1}^N$, and define  an inner product $(\cdot, \cdot )_K$ in $\bm H' \otimes \mathbb C^N$ such that
\begin{equation}
\left( \sum_k | \varphi_k \rangle | e_k \rangle, \sum_k | \psi_k \rangle | e_k \rangle \right)_K := \sum_k \langle \varphi_k | X_k | \psi_k \rangle.
\end{equation}
By using the Schwarz inequality for the inner product $(\cdot, \cdot )_K$, we  have
\begin{equation}
\begin{split}
\langle z | z \rangle &= \sum_k \langle z | x_k X_k | \psi \rangle \\ 
&=  \left( \sum_k | z \rangle | e_k \rangle, \sum_k x_k | \psi \rangle | e_k \rangle \right)_K \\
&\leq \left( \sum_k | z \rangle | e_k \rangle, \sum_k |z \rangle | e_k \rangle \right)_K^{1/2} \left( \sum_k x_k | \psi \rangle | e_k \rangle, \sum_k x_k | \psi \rangle | e_k \rangle \right)_K^{1/2} \\
&= \left( \sum_k \langle z |X_k | z \rangle \right)^{1/2}  \left( \sum_k \langle \psi |  x_k X_k x_k | \psi \rangle \right)^{1/2} \\
&= \langle z |  z \rangle^{1/2}  \left( \sum_k \langle \psi |  x_k^2 X_k | \psi \rangle \right)^{1/2},
\end{split}
\end{equation}
which implies inequality~(\ref{Kadison1}).
$\Box$\end{proof}

We note that we did not assume the complete positivity of $\mathcal E$ for the Kadison inequality.  The following lemma is a straightforward consequence:

\begin{lemma}[Schwarz's operator inequality~\cite{Choi2}]
Let $X \in L(\bm H)$ be an arbitrary  operator and $\mathcal E: L(\bm H) \to L(\bm H')$ be a unital $2$-positive map.  Then
\begin{equation}
\mathcal E (X^\dagger X) \geq \mathcal E(X^\dagger) \mathcal E (X).
\label{Schwarz}
\end{equation}
\end{lemma}

\begin{proof}
By applying the Kadison inequality~(\ref{Kadison}) to positive map $\mathcal E \otimes \mathcal I_2$ and a Hermitian operator 
\begin{equation}
X' := \left[
\begin{array}{cc}
0 & X^\dagger \\
X & 0 \\
\end{array}
\right] \in L(\bm H \otimes \mathbb C^2 ),
\end{equation}
we obtain $(\mathcal E \otimes \mathcal I_2)  (X'^2 ) \geq  \left( (\mathcal E \otimes \mathcal I_2)  (X')\right)^2$, or equivalently
\begin{equation}
\left[
\begin{array}{cc}
\mathcal E (X^\dagger X) & 0 \\
0 & \mathcal E (XX^\dagger) \\
\end{array}
\right] \geq \left[ 
\begin{array}{cc}
\mathcal E(X^\dagger) \mathcal  E(X) & 0 \\
0 & \mathcal E(X) \mathcal E (X^\dagger) \\
\end{array}
\right],
\end{equation} 
which implies inequality (\ref{Schwarz}).
$\Box$\end{proof}

\

We now prove the monotonicity~(\ref{monotonicity}) in~Theorem 3.5.
Let $\mathcal E : L(\bm H) \to L(\bm H')$ be a CPTP map and  $\rho, \sigma \in Q(\bm H)$ be states.
For simplicity, we assume that $\rho$, $\sigma$, $\mathcal E (\rho)$, and $\mathcal E(\sigma)$ are positive definite.
We define $\mathcal L$, $\mathcal R$, and $\mathcal D$ that act on $L(\bm H)$ such that, for $X \in L(\bm H)$,
\begin{equation}
\mathcal L (X) := \sigma X, \ \mathcal R (X) := X\rho^{-1}, \ \mathcal D (X) :=  \sigma X \rho^{-1},
\end{equation}
where $\mathcal D = \mathcal L \mathcal R = \mathcal R \mathcal L$.  
We note that $\mathcal L$, $\mathcal R$, and $\mathcal D$ are positive definite in terms of the Hilbert-Schmidt inner product~(\ref{Hilbert_Schmidt}).

Let $\sigma := \sum_k p_k | \varphi_k \rangle \langle \varphi_k |$ be the spectrum decomposition.  Then the eigenvectors of $\mathcal L$ are  $\{ | \varphi_k \rangle \langle \varphi_l | \}_{kl}$ and the eigenvalues are $\{ p_k \}_k$.  We then have $(\ln \mathcal L) (X) = (\ln \sigma)X$. Similarly, $(\ln \mathcal R) (X) = - X(\ln \rho)$.  Therefore,
\begin{equation}
(\ln \mathcal D) (X) = (\ln \mathcal L + \ln \mathcal R) (X) = (\ln \sigma)X - X(\ln \rho).
\end{equation}
 By using the Hilbert-Schmidt inner product, we have
\begin{equation}
\begin{split}
S(\rho \| \sigma) &= \langle \rho^{1/2}, (\ln \rho) \rho^{1/2} \rangle_{\rm HS} - \langle \rho^{1/2}, (\ln \sigma) \rho^{1/2}\rangle_{\rm HS} \\
&= \langle \rho^{1/2}, (- \ln \mathcal D) ( \rho^{1/2} ) \rangle_{\rm HS}.
\end{split}
\end{equation}
On the other hand,
\begin{equation}
S(\mathcal E (\rho) \| \mathcal E (\sigma))  = \langle \mathcal E(\rho)^{1/2}, (- \ln \mathcal D') (\mathcal E(\rho)^{1/2}) \rangle_{\rm HS},
\end{equation}
where $\mathcal D'$ acts on $L(\bm H')$ such that
\begin{equation}
\mathcal D' (X) :=  \mathcal E(\sigma) X  \mathcal  E(\rho)^{-1} \ (X \in L(\bm H')).
\end{equation}
By noting that
\begin{equation}
-\ln x = \int_0^{+\infty} \left[ f_t(x)  + ( t(1+t) )^{-1} \right] dt,
\end{equation}
where $f_t (x) := (x+t)^{-1} - t^{-1}$ ($0 < t < + \infty$), we have
\begin{equation}
S(\rho \| \sigma) = \int_0^{+\infty} \left[ \langle \rho^{1/2}, f_t (\mathcal D) ( \rho^{1/2} ) \rangle_{\rm HS} + ( t(1+t) )^{-1} \right] dt,
\end{equation}
\begin{equation}
S(\mathcal E (\rho) \| \mathcal E (\sigma))  = \int_0^{+\infty} \left[ \langle \mathcal E (\rho )^{1/2}, f_t (\mathcal D') ( \mathcal E (\rho )^{1/2} ) \rangle_{\rm HS} + ( t(1+t) )^{-1} \right] dt.
\end{equation}
Therefore, in order to to prove the monotonicity~(\ref{monotonicity}), it is sufficient to show that for any $0 < t < + \infty$ 
\begin{equation}
\langle \rho^{1/2}, f_t (\mathcal D) ( \rho^{1/2} ) \rangle_{\rm HS} \geq \langle \mathcal E (\rho )^{1/2}, f_t (\mathcal D') ( \mathcal E (\rho )^{1/2} ) \rangle_{\rm HS}.
\end{equation}

We next define $\mathcal V : L(\bm H') \to L(\bm H)$ such that
\begin{equation}
\mathcal V (X) := \mathcal E^\dagger \left( X \mathcal E(\rho)^{-1/2} \right) \rho^{1/2},
\end{equation}
or equivalently
\begin{equation}
\mathcal V \left( X \mathcal E(\rho)^{1/2} \right) = \mathcal E^\dagger (X) \rho^{1/2}
\end{equation}
for  $X \in L(\bm H')$.
We note that  $\mathcal V (\mathcal E (\rho)^{1/2} ) = \rho^{1/2}$ holds, because $\mathcal E$ is trace-preserving so that $\mathcal E^\dagger$ is unital.  We then have
\begin{equation}
\begin{split}
\langle \rho^{1/2}, f_t (\mathcal D) \rho^{1/2}\rangle_{\rm HS} &=  \langle \mathcal V (\mathcal E (\rho)^{1/2} ), f_t (\mathcal D) \mathcal V (\mathcal E (\rho)^{1/2} ) \rangle_{\rm HS} \\
&=  \langle \mathcal E (\rho)^{1/2}, \mathcal V^\dagger f_t (\mathcal D) \mathcal V (\mathcal E (\rho)^{1/2} ) \rangle_{\rm HS}.
\end{split}
\end{equation}
Therefore, our goal is to show that
\begin{equation}
\langle \mathcal E (\rho)^{1/2}, \mathcal V^\dagger f_t (\mathcal D) \mathcal V (\mathcal E (\rho)^{1/2} ) \rangle_{\rm HS} \geq \langle \mathcal E(\rho)^{1/2}, f_t (\mathcal D') (\mathcal E(\rho)^{1/2}) \rangle_{\rm HS}.
\label{goal1}
\end{equation}
To show inequality~(\ref{goal1}), it is sufficient to show that
\begin{equation}
\mathcal V^\dagger f_t (\mathcal D) \mathcal V \geq f_t (\mathcal D').
\label{goal2}
\end{equation}

$\mathcal V$ is a contraction in terms of the Hilbert-Schmidt inner product, because
\begin{equation}
\begin{split}
&\langle \mathcal E^\dagger (X) \rho^{1/2}, \mathcal E^\dagger (X) \rho^{1/2} \rangle_{\rm HS} = {\rm tr} \left[ \rho \mathcal E^\dagger (X^\dagger ) \mathcal E^\dagger (X) \right] \\
&\leq {\rm tr} \left[  \rho \mathcal E^\dagger (X^\dagger X)  \right] = {\rm tr} \left[ \mathcal E (\rho) X^\dagger X \right] \\
&= \langle  X \mathcal E (\rho)^{1/2}, X \mathcal E (\rho)^{1/2} \rangle_{\rm HS},
\end{split}
\end{equation}
where we used the Schwarz inequality~(\ref{Schwarz}) in Lemma~B.6 for $\mathcal E^\dagger$ that is unital.  Since $f_t(x)$ is operator convex from Lemma~B.2 and $\lim_{x \to + 0} f_t (x) = 0$, we obtain
\begin{equation}
f_t (\mathcal V^\dagger \mathcal D \mathcal V) \leq  \mathcal V^\dagger f_t ( \mathcal D ) \mathcal V,
\label{appendix_inequality1}
\end{equation}
where we used the Jensen's operator inequality~(\ref{operator_convex2}) in Lemma~B.4.

We next show that
\begin{equation}
\mathcal V^\dagger \mathcal D \mathcal V \leq \mathcal D'.
\label{appendix_inequality2}
\end{equation}
In fact,
\begin{equation}
\begin{split}
&\langle X \mathcal E(\rho)^{1/2}, \mathcal V^\dagger \mathcal D \mathcal V ( X \mathcal E(\rho)^{1/2}) \rangle_{\rm HS} \\
&= \langle \mathcal V (X \mathcal E(\rho)^{1/2} ), \mathcal D \mathcal V ( X \mathcal E(\rho)^{1/2}) \rangle_{\rm HS} \\
&= \langle \mathcal E^\dagger (X) \rho^{1/2}, \mathcal D (\mathcal E^\dagger (X) \rho^{1/2} ) \rangle_{\rm HS}
= {\rm tr} \left[ \rho \mathcal E^\dagger (X^\dagger) \mathcal D \mathcal E^\dagger (X) \right]\\
&= {\rm tr} \left[ \sigma \mathcal E^\dagger (X) \mathcal E^\dagger (X^\dagger )  \right]
\leq  {\rm tr} \left[ \sigma \mathcal E^\dagger (X X^\dagger )  \right] = {\rm tr} \left[ \mathcal E (\sigma)X X^\dagger   \right]  \\
&= {\rm tr} [\mathcal E (\rho)^{1/2} X^\dagger \mathcal E(\sigma) X \mathcal E(\rho)^{1/2} \mathcal E(\rho)^{-1} ] \\
&= \langle X \mathcal E(\rho)^{1/2}, \mathcal D' ( X \mathcal E(\rho)^{1/2}) \rangle_{\rm HS},
\end{split}
\end{equation}
where we again used the Schwarz inequality~(\ref{Schwarz}) for $\mathcal E^\dagger$.
Since $f_t (x)$ is decreasing-operator monotone from Lemma~B.1, we obtain
\begin{equation}
f_t ( \mathcal D' ) \leq  f_t (\mathcal V^\dagger \mathcal D \mathcal V).
\label{appendix_inequality3}
\end{equation}

By combining inequalities~(\ref{appendix_inequality1}) and (\ref{appendix_inequality3}), we finally obtain inequality~(\ref{goal2}), which implies the monotonicity~(\ref{monotonicity}). 
We note that the assumption of the complete positivity has been used only for the proof of  the Schwarz's operator inequality~(\ref{Schwarz}), in which the assumption of the $2$-positivity is in fact enough.

\

\section*{Acknowledgments}
The author thanks to Yu Watanabe for valuable comments. 
The author also thanks to Hal Tasaki and Kouki Nakata for valuable comments for the revision of the manuscript.
This work was supported by the Grant-in-Aid for Research Activity Start-up (KAKENHI 11025807).


\end{document}